\documentclass[showpacs,amsmath,amsfonts,amssymb,aps,superscriptaddress]{revtex4}
\usepackage{amsmath,amssymb}
\usepackage[utf8]{inputenc}
\usepackage[T1]{fontenc}
\usepackage{microtype}
\usepackage{caption}
\usepackage{parskip}
\usepackage{physics}
\usepackage{xcolor}
\newtheorem{theorem}{Theorem}[section]
\newtheorem{lemma}[theorem]{Lemma}

\newenvironment{proof}[1][Proof.]{\begin{trivlist}
		\item[\hskip \labelsep {\bfseries #1}]}{\end{trivlist}}

\usepackage{mathtools}
\usepackage{listings}
\usepackage{booktabs}
\usepackage{float}
\usepackage{graphicx}
\usepackage{subcaption}

\usepackage{natbib}
\begin{document}

    \title{An exact Kerr-like rotating Morris--Thorne wormhole: throat, ergoregion and equatorial shadow slice}
    
    \author{Mark Essa Sukaiti}
    \email{100064482@ku.ac.ae}
    \affiliation{Mathematics Department, Khalifa University of Science and Technology, PO Box 127788, Abu Dhabi, United Arab Emirates}
    
    \author{Davide Batic}
    \email{davide.batic@ku.ac.ae}
    \affiliation{Mathematics Department, Khalifa University of Science and Technology, PO Box 127788, Abu Dhabi, United Arab Emirates}
    
    \author{Denys Dutykh}
    \email{denys.dutykh@ku.ac.ae}
    \affiliation{Mathematics Department, Khalifa University of Science and Technology, PO Box 127788, Abu Dhabi, United Arab Emirates}
    
    \date{\today}

\begin{abstract}
We construct an exact Kerr-like rotating extension of the zero-redshift Morris--Thorne wormhole within a Teo-type stationary and axisymmetric ansatz. Imposing the Einstein equations for an anisotropic source comoving with zero angular momentum observers fixes the circumferential metric function to $\mathcal{K}(r)=r^2+a^2$ and reduces the frame-dragging function to a single radial quadrature. For $b(r)=r_0^2/r$, this quadrature is evaluated in terms of incomplete elliptic integrals and normalised by the ADM angular momentum. The spacetime is asymptotically flat but has vanishing ADM mass, so the Kerr relation $a=J_{\rm ADM}/M_{\rm ADM}$ is not applicable. The field equations leave the Kerr-like oblateness length $a$ and the physical angular momentum $J_{\rm ADM}$ independent. The exact family is therefore labelled by $(r_0,a,J_{\rm ADM})$. The relation $J_{\rm ADM}=ar_0$ is used only to select a one-parameter slice for the numerical ergoregion and equatorial-capture illustrations, and not as a field-equation constraint. We obtain $r_{\rm th}=\sqrt{r_0^2-a^2}$, the reality bound $|a|\leq r_0$, and the regular range $0\leq |a|<r_0$. Within the canonical stationary foliation, the throat is characterised quasi-locally as the unique member $S_0$ of the closed two-surface family $S_\ell$ with vanishing mean curvature and positive area second variation. $r=r_{\rm th}$ and $\ell=0$ are only coordinate representations of this surface. A signed throat-adapted coordinate displays two asymptotically flat ends, excludes closed timelike curves, and shows that the throat is timelike rather than a horizon. These geometrical and causal conclusions do not require $J_{\rm ADM}=ar_0$. The supporting anisotropic stress tensor is reconstructed from the Einstein tensor and is interpreted here as an effective phenomenological source. No microscopic matter Lagrangian or realistic equation of state is claimed. We also determine the ergoregion onset. Because the full Hamilton--Jacobi equation is not separable, the optical result is only the one-dimensional equatorial capture interval, not the complete two-dimensional shadow.
\end{abstract}

\maketitle

\section{Introduction}

Wormholes provide a useful setting for exploring the relation between spacetime topology, the gravitational field equations, and the matter content required to support nontrivial geometries. In their simplest form, they describe geometries with a throat connecting two asymptotic regions or two distant domains of the same spacetime. Since the Morris--Thorne construction and the earlier Ellis--Bronnikov solutions, static wormholes have served as standard laboratories for studying the geometrical requirements of traversability, the role of energy conditions, and the possible signatures of horizonless compact objects \cite{Morris1988AJP, Ellis1973JMP, Ellis1974JMP, Ellis1979GRG, Bronnikov1973APP, Visser1996, Lobo2008CQGR}. They have also been discussed in modified theories of gravity and in effective or quantum-inspired settings, where the stress-energy responsible for the throat may be partly geometrical or effective rather than purely material \cite{Bhawal1992PRD, Kanti2011PRL, Kanti2012PRD, Antoniou2020PRD, Gao2017JHEP}. From this point of view, exact solutions remain valuable because they make it possible to distinguish coordinate effects from invariant geometrical properties and to trace explicitly how the matter source enters the Einstein equations. For static, spherically symmetric wormholes in general relativity, the geometry is usually written in terms of a redshift function $\Phi(r)$ and a shape function $b(r)$. The existence of a throat imposes a flare-out condition, and in ordinary Einstein gravity, this is typically associated with a violation of the null energy condition near the throat \cite{Morris1988AJP, Visser1996, Lobo2008CQGR, Wall2013CQG}. For a nonspherical throat, the appropriate statement is quasi-local and geometrical, that is on a chosen spacelike hypersurface the throat is a compact two-surface of vanishing mean curvature with a positive flare-out or second-variation condition. A related spacetime formulation uses the expansions of the orthogonal null congruences \cite{HochbergVisser1997PRD,HochbergVisser1998PRD,TomikawaIzumiShiromizu2015PRD}. This feature is not merely a technical detail because the distinction between physical stress-energy and effective stress-energy is central to any assessment of a wormhole model. Consequently, when rotating generalisations are considered, one must keep track not only of the metric functions but also of the source that supports them. In particular, a rotating metric obtained from a prescribed ansatz or a solution-generating prescription need not automatically have a transparent matter interpretation. Rotation is nevertheless an essential ingredient in any comparison with compact astrophysical objects. Stationary and axisymmetric wormholes may exhibit frame dragging, ergoregions, modified photon dynamics, and optical features that differ from those of Kerr black holes. Rotation also makes the field equations considerably more restrictive because the lapse, shift, circumferential radius, shape function, and stress-energy components are coupled by consistency conditions that are absent in the static case. Teo's rotating traversable wormhole metric provided an important phenomenological framework for studying such effects \cite{Teo1998PRD}. Subsequent work has explored slowly rotating wormholes, numerical rotating solutions, and rotating configurations supported by specific matter sectors or by modified gravitational dynamics \cite{Kuhfittig2003PRD, Kashargin2008GC, Kashargin2008PRD, Kleihaus2014PRD, Hoffmann2018PRD, Chew2019PRD, Azad2023, Azad2024PLB, Cisterna2023PRD, Clement2023PLB, Tanghpati2024NPB, Batic2026EPJC, Batic2026CQGSlow}. In parallel, rotating wormholes have been studied as black-hole mimickers through lensing, accretion, and shadow observables \cite{Bambi2013PRD, Bambi2013PRDa, Nedkova2013PRD, Gyulchev2018EPJC, Shaikh2018PRD, Deligianni2021PRD, Deligianni2021PRDa, Kumar2024PRD}. A closely related benchmark is the exact spinning Morris--Thorne wormhole constructed in \cite{Batic2026EPJC}. In that work, the rotating geometry was obtained within a Teo-type ansatz with unit lapse and a Morris--Thorne shape function, yielding an analytic frame-dragging function and a two-parameter asymptotically flat family labelled by the throat radius and the total angular momentum. The analysis of \cite{Batic2026EPJC} included curvature invariants, stress-energy components, energy condition violations, causal structure, shadows, and Geroch--Hansen multipole moments. The present paper is complementary to that construction. Instead of preserving the same spherical throat structure, we introduce a Kerr-like oblate deformation through the functions $\Sigma=r^2+a^2\cos^2{\vartheta}$, $\Delta=r^2+a^2-rb(r)$ where the parameter $a$ enters the metric functions as an oblateness parameter. The physical angular momentum is then fixed asymptotically by the falloff of $g_{t\varphi}$, and it should not be identified with $M_{\rm ADM}a$ because the solution considered below is massless in the ADM (Arnowitt, Deser, and Misner) sense. More precisely, before any optional phenomenological specialisation, $a$ is an oblateness length entering the metric through $a^2$, whereas the sign and magnitude of the rotational dragging are carried by the independently measurable quantity $J_{\rm ADM}$.\\
In this work, we construct an exact rotating extension of the zero-redshift Morris--Thorne wormhole in general relativity. We start with a modified Teo-type stationary and axisymmetric ansatz and express the source as an anisotropic stress tensor in an orthonormal frame comoving with zero-angular-momentum observers. After imposing $N=1$, $\mathcal K=\mathcal K(r)$, and $\omega=\omega(r)$, the Einstein equations and the consistency of the azimuthal pressure equations fix $\mathcal K(r)=r^2+a^2$ and reduce the frame-dragging function to a single radial quadrature. For $b(r)=r_0^2/r$, this quadrature is evaluated in terms of incomplete elliptic integrals and its amplitude is fixed by the independently defined ADM angular momentum $J_{\rm ADM}$. Since $M_{\rm ADM}=0$, the Kerr relation $a=J_{\rm ADM}/M_{\rm ADM}$ is inapplicable, and no remaining field equation relates $a$ to $J_{\rm ADM}$. We therefore introduce the independent dimensionless parameters $\hat a=a/r_0$ and $\hat J=J_{\rm ADM}/r_0^2$. The prescription $J_{\rm ADM}=ar_0$, equivalently $\hat J=\hat a$, is retained only as an explicitly identified one-parameter slice for the numerical plots and tabulated optical observables. The Morris--Thorne static seed is reached in the joint limit $(a,J_{\rm ADM})\to(0,0)$. Notice that the linear prescription is one convenient path to that limit, but it is not unique. 
This separation also makes the logical dependence of the results transparent. The exact dragging profile and the source components depend on $a$ and $J_{\rm ADM}$ separately, whereas the throat geometry, regularity bound, and causal extension do not require the relation $J_{\rm ADM}=ar_0$. Most importantly, within the natural stationary foliation defined by the asymptotically normalised time coordinate and used throughout, the throat is not defined by the coordinate equation $r=\sqrt{r_0^2-a^2}$. It is the unique member $S_0$ of the closed two-surface family $S_\ell$ on which the mean curvature vanishes, the area has a strict local minimum, and both future null expansions vanish. These statements are invariant under reparametrisations of the radial coordinate. The equations $\ell=0$, $\Delta=0$, and $r=r_{\rm th}$ are merely chart-dependent descriptions of that same geometrical surface. Its equatorial proper circumference is $2\pi r_0$, and the regular family has $0\leq |a|<r_0$. The ergoregion threshold is a curve in the $(\hat a,\hat J)$ plane, and the quoted number $0.2087018708$ is the intersection of that curve with the illustrative slice $\hat J=\hat a$. The analytic optical endpoints are likewise derived for independent $a$ and $J_{\rm ADM}$, but they determine only the equatorial capture set. For an equatorial observer this is the horizontal celestial slice $\mathcal S_{\rm eq}=\mathcal S\cap\{\beta=0\}$ of the full two-dimensional shadow set $\mathcal S$. The numerical table and figures represent only that horizontal slice and not an off-equatorial shadow contour. In the signed throat coordinate, $g^{tt}=-1$, the throat world tube is timelike, and closed timelike curves are excluded. We use traversability only in the limited sense of causal passability by test particles. Furthermore, we do not claim to have constructed the first exact rotating wormhole. The Teo ansatz, rotating phantom-matter solutions, solution-generating constructions, Kerr-like black-bounce geometries, and optical separatrix methods all predate the present work. The role of the following comparison is to state precisely what is inherited and what is new within the restricted ansatz $N=1$, $\mathcal K=\mathcal K(r)$, and $\omega=\omega(r)$.
\begin{table*}[t]
\centering
\caption{Representative comparison with previous rotating-wormhole constructions. ``New here'' always means within the assumptions of the present Kerr-like zero-redshift Morris--Thorne ansatz. No priority claim is made for rotating wormholes in general.}
\label{tab:literature-comparison}
\scriptsize
\renewcommand{\arraystretch}{1.24}
\begin{tabular}{p{0.17\textwidth}p{0.30\textwidth}p{0.35\textwidth}}
\toprule
Representative literature & Established construction or method & Distinct result of the present paper \\
\midrule
Teo metrics and optical analyses \cite{Teo1998PRD,Shaikh2018PRD,ChengXuZhao2026EPJC} & Stationary axisymmetric wormhole ansatz with prescribed metric functions; exterior photon orbits and throat-critical rays as optical separatrices. & The Einstein equations and source consistency determine $\mathcal K=r^2+a^2$ and the radial quadrature for $\omega$; for $b=r_0^2/r$ the latter is evaluated in elliptic integrals. The resulting optical calculation is explicitly done only on an equatorial slice. \\
Rotating Ellis/phantom and explicit-matter models \cite{Kashargin2008PRD,Kleihaus2014PRD,Hoffmann2018PRD,Chew2019PRD,DzhunushalievFolomeev2026GRG,Batic2026CQGSlow} & Matter sectors are specified from the outset; many solutions are perturbative or numerical, and their matter interpretation is correspondingly sharper. & The present geometry is exact and analytic but reconstructs only an effective stress tensor. It is therefore complementary, not a microscopic replacement for those matter-supported solutions. \\
Exact solution-generated and Kerr-like families \cite{Azreg2014EPJC,Mazza2021JCAP,KarJanaKar2025PRD,Cisterna2023PRD,Clement2023PLB,Tanghpati2024NPB} & Exact rotating regular, black-bounce, or wormhole geometries obtained from different static seeds, Newman--Janis-type prescriptions, Ehlers/Harrison transformations, or specified electromagnetic and higher-form sectors. & Here the seed is the zero-redshift Morris--Thorne geometry with $b=r_0^2/r$; the ADM mass vanishes, $a$ and $J_{\rm ADM}$ remain independent, and the oblate throat is characterised by minimal-area and null-expansion conditions. \\
Exact spinning Morris--Thorne solution of Ref.~\cite{Batic2026EPJC} & Spherical constant-radius two-geometries, parameters $(r_0,J)$, an elementary frame-dragging profile, separable null dynamics with full shadow contours, and Geroch--Hansen multipoles. & The present family adds an independent oblateness scale $a$, an oblate minimal throat, an elliptic-integral dragging profile, and a two-parameter ergoregion curve. The price of the oblate deformation is loss of full Hamilton--Jacobi separability, so only the equatorial shadow slice is obtained analytically. \\
\bottomrule
\end{tabular}
\end{table*}
Accordingly, the new results claimed here are: the field-equation determination of $\mathcal K=r^2+a^2$ within the stated ansatz, the closed elliptic-integral dragging function for the Morris--Thorne shape function, the independent $(a,J_{\rm ADM})$ parameter structure, the quasi-local minimal and marginal character of the oblate throat together with its regular two-ended extension, the general two-parameter ergoregion onset, and the closed analytic endpoints of the equatorial capture interval. The use of a Teo-type ansatz, minimal-surface flare-out criteria, ergoregion diagnostics, and photon separatrices extends established methods rather than constituting new general principles.\\
The paper is organised as follows. In Sec.~II we introduce the static seed geometry, the Kerr-like Teo ansatz, and the reduced Einstein equations. Moreover, we derive the exact solution with independent $a$ and $J_{\rm ADM}$, and discuss the status of the reconstructed source. In Sec.~III we analyse the throat geometry and causal extension, derive the general two-parameter ergoregion condition and equatorial optical boundaries, and then present the one-parameter numerical illustrations. The final section summarises the parameter dependence and the limitations of the effective matter interpretation.

\section{Derivation of the Kerr-like Morris--Thorne wormhole}

As the wormhole seed metric, \emph{i.e.}, in the absence of rotation, we consider a static, spherically symmetric spacetime. In natural units $c = G_N = 1$, the line element of such a wormhole can be written as \cite{Morris1988AJP, Ellis1973JMP, Ellis1979GRG, Ellis1974JMP, Bronnikov1973APP}
\begin{equation}\label{metric}
  ds^2 = -e^{2\Phi(r)}dt^2+\frac{dr^2}{1-\frac{b(r)}{r}} + \frac{r^2 d\chi^2}{1-\chi^2} + r^2(1-\chi^2)d\varphi^2, \quad \chi\in[-1,1], \quad \varphi\in[0,2\pi).
\end{equation}
Here, $\Phi$ and $b$ are the redshift and shape functions, respectively. We assume that the radial coordinate $r$ increases monotonically from its minimum value $r_0$, representing the throat of the wormhole, to spatial infinity. The matter content acting as a source of the geometry described by \eqref{metric} is modelled in terms of an anisotropic fluid, with energy-momentum tensor given by
\begin{equation}\label{emt}
 T^\alpha{}_\beta=(\rho+p_t)u^\alpha u_\beta+p_t\delta^\alpha{}_\beta+(p_r-p_t)\ell^\alpha\ell_\beta,    
\end{equation}
where $\ell^\alpha$ is a unit spacelike vector orthogonal to the fluid four-velocity $u^\alpha$, that is $\ell^\alpha \ell_\alpha=1$ and $\ell^\alpha u_\alpha=0$. Moreover, $u^\alpha$ must satisfy the condition $g_{\alpha\beta}u^\alpha u^\beta=-1$. These constraints require that
\begin{equation}\label{utnorot}
    u^\alpha=e^{-\Phi(r)}\delta^\alpha{}_t,\quad
    \ell^\alpha=\sqrt{1-\frac{b(r)}{r}}\delta^\alpha{}_r.
\end{equation}
Hence, the mixed energy-momentum tensor is $T^\alpha{}_\beta=\text{diag}(-\rho(r), p_r(r), p_t(r), p_t(r))$, where $\rho$ is the energy density, $p_r$ the radial pressure, and $p_t$ the tangential pressure measured orthogonally to the radial direction. By applying the Einstein field equations $G_{\alpha\beta}=8\pi T_{\alpha\beta}$ alongside the conservation equation $\nabla_\alpha T^{\alpha\beta}=0$, we obtain the following system of equations, where an overdot denotes differentiation with respect to the radial coordinate
\begin{eqnarray}
&&\dot{b}-8\pi r^2\rho=0,\label{eq1}\\
&&2r(r-b)\dot{\Phi}-b-8\pi r^3 p_r=0,\label{eq2}\\
&&r^2(r-b)\ddot{\Phi}+\left[r(r-b)\dot{\Phi}+\frac{b-r\dot{b}}{2}\right](1+r\dot{\Phi})-8\pi r^3 p_t=0,\label{eq3}\\
&&r(\rho+p_r)\dot{\Phi}+r\dot{p}_r+2(p_r-p_t)=0.\label{eq4}
\end{eqnarray}
Using equations \eqref{eq1}--\eqref{eq3}, we can express the energy density, the radial and tangential pressures in terms of the redshift and shape functions as follows
\begin{eqnarray}
\rho&=&\frac{\dot{b}}{8\pi r^2},\label{ed}\\
p_r&=&\frac{1}{8\pi}\left[\frac{2}{r}\left(1-\frac{b}{r}\right)\dot{\Phi}-\frac{b}{r^3}\right],\label{pr}\\
p_t&=&\frac{1}{8\pi}\left(1-\frac{b}{r}\right)\left[\ddot{\Phi}+\dot{\Phi}^2-\frac{r\dot{b}-b}{2r(r-b)}\dot{\Phi}+\frac{\dot{\Phi}}{r}-\frac{r\dot{b}-b}{2r^2(r-b)}\right].\label{pt}
\end{eqnarray}
As a rotating extension of the static wormhole metric \eqref{metric}, we propose a modified Teo ansatz \cite{Teo1998PRD} written in Boyer-Lindquist–type coordinates $(t,r,\chi,\varphi)$ with $\chi=\cos{\theta}$
\begin{equation}\label{rotmet}
  ds^2=-[N^2-(1-\chi^2) \mathcal{K} \omega^2]dt^2+\frac{\Sigma}{\Delta}dr^2-2\omega(1-\chi^2)\mathcal{K} dtd\varphi+\frac{\Sigma }{1-\chi^2} d\chi^2+(1-\chi^2) \mathcal{K} d\varphi^2
\end{equation}
with $\Sigma=r^2+a^2 \chi^2 $ and $\Delta= r^2 + a^2 - r b(r)$. The spatial variables $(r,\theta,\varphi)$ should be understood as oblate spheroidal coordinates, with $a$ playing the role of the oblateness parameter. Here $N(r,\chi)$ is the lapse, $\omega(r,\chi)$ is the angular velocity of inertial frames, and $\mathcal{K}(r,\chi)$ determines the circumferential radius associated with the azimuthal direction. The static seed is obtained in the joint nonrotating and non-oblate limit, for which $a\to0$ and the independently normalised dragging amplitude tends to zero, equivalently $J_{\rm ADM}\to 0$
\begin{equation}
N(r,\chi)\to e^{\Phi(r)},\qquad
\mathcal{K}(r,\chi)\to r^2,\qquad
\omega(r,\chi)\to 0.
\end{equation}
Thus, $a$ has dimensions of length and is a Kerr-like oblateness parameter in $\Sigma$ and $\Delta$. It is not, by itself, the physical angular momentum parameter. The physical ADM angular momentum, however, is determined by the asymptotic falloff of the off-diagonal metric component. Since $g_{t\varphi}=-(1-\chi^2)\mathcal{K}(r,\chi)\omega(r,\chi)$, and $\mathcal{K}(r)\sim r^2$ as $r\to\infty$, asymptotic flatness with angular momentum $J_{\rm ADM}$ requires
\begin{equation}
g_{t\varphi}\sim -\frac{2J_{\rm ADM}}{r}(1-\chi^2),
\end{equation}
or equivalently
\begin{equation}
\omega(r,\chi)\sim \frac{2J_{\rm ADM}}{r^3}.
\end{equation}
Thus $J_{\rm ADM}$ is fixed by the $1/r^3$ coefficient of $\omega$, while $a$ is a separate oblateness length. Only in a Kerr-like spacetime with nonzero ADM mass $M_{\rm ADM}$ may one further identify
\begin{equation}\label{Kerr-rel}
  a=\frac{J_{\rm ADM}}{M_{\rm ADM}}.
\end{equation}
We parameterise the stress tensor required by the geometry in the algebraic form of a type-I anisotropic source. This tetrad decomposition should not be confused with a microscopic fluid model. Its physical status is discussed in Sec.~\ref{subsec:source-status}. Let $\{U^{\alpha}, n_{(r)}^{\alpha}, n_{(\chi)}^{\alpha}, n_{(\varphi)}^{\alpha}\}$ be an orthonormal tetrad comoving with the effective source, with $U_{\alpha}U^{\alpha}=-1$, $n_{(i)\,\alpha}n_{(j)}^{\alpha}=\delta_{ij}$, and $U_{\alpha}n_{(i)}^{\alpha}=0$. In its covariant form, the energy–momentum tensor is
\begin{equation}\label{AEMT}
  \widetilde{T}_{\alpha\beta}=\widetilde{\rho}U_\alpha U_\beta + P_r n_{(r)\alpha}n_{(r)\beta} + P_\chi n_{(\chi)\alpha}n_{(\chi)\beta} + P_\varphi n_{(\varphi)\alpha}n_{(\varphi)\beta},    
\end{equation}
where $\widetilde{\rho}(r,\chi)$ is the energy density measured by comoving observers, and $P_{r}(r,\chi), P_{\chi}(r,\chi), P_{\varphi}(r,\chi)$ are the principal pressures along the radial, polar, and azimuthal directions, respectively. The four-velocity is chosen to describe a ZAMO flow, \emph{i.e.} a congruence corotating with the local inertial frames. By definition, a ZAMO (Zero Angular Momentum Observer) has zero angular momentum, \emph{i.e.} $U_\varphi=g_{\varphi\varphi}(\Omega-\omega)U^t=0$, which implies that its angular velocity satisfies the condition $\Omega = \omega$. Accordingly, we take $U^{\alpha} = N^{-1}(1,0,0,\omega)$, with the normalisation fixed by $U_{\alpha}U^{\alpha}=-1$.  Furthermore, if we impose $n_{(i)\,\alpha}n_{(j)}^{\alpha}=\delta_{ij}$, and $U_{\alpha}n_{(i)}^{\alpha}=0$, a straightforward computation shows that
\begin{equation}
  n^{\alpha}_{(r)}=\left(0,\sqrt{\frac{\Delta}{\Sigma}},0,0\right),\quad
  n^{\alpha}_{(\chi)}=\left(0,0,\sqrt{\frac{1-\chi^2}{\Sigma}},0\right),\quad
  n^{\alpha}_{(\varphi)}=\left(0,0,0,\frac{1}{\sqrt{(1-\chi^2) \mathcal{K}}}\right).
\end{equation}
Finally, the non-trivial Einstein field equations are
\begin{eqnarray}
G_{tt}&=& 8\pi N^2\widetilde{\rho}+8\pi(1-\chi^2)\omega^2 \mathcal{K} P_\varphi\\
G_{t\varphi}&=&-8\pi(1-\chi^2)\omega \mathcal{K} P_\varphi,\label{Gtphi}\\
G_{\varphi\varphi}&=&8\pi(1-\chi^2) \mathcal{K} P_\varphi,\label{Gphiphi}\\
G_{rr}&=&\frac{8\pi \Sigma}{\Delta}P_r,\label{Grr}\\
G_{\chi\chi}&=&\frac{8\pi\Sigma }{(1-\chi^2)}P_\chi,\label{Gchichi}\\
G_{r\chi}&=&0.\label{Grchi}
\end{eqnarray}
Equation \eqref{Grchi} can be expressed as
\begin{equation}\label{cond1}
Q_1(N,\mathcal{K})\partial_\chi\Sigma+Q_2(N,\mathcal{K}) \partial_r\Sigma+Q_3(N,\mathcal{K},\omega)\Sigma=0    
\end{equation}
with
\begin{eqnarray}
Q_1 (N,\mathcal{K}) &=& - \mathcal{K}(1-\chi^2)\,\partial_r(N^2\mathcal{K}), \\
Q_2 (N, \mathcal{K}) &=& -\mathcal{K}\,\partial_\chi\!\left[(1-\chi^2)N^2\mathcal{K}\right], \\
Q_3 (N, \mathcal{K}, \omega) &=& 2\mathcal{K}(1-\chi^2) \left[ N^2\,\partial_{r\chi}\mathcal{K} +2N\mathcal{K}\,\partial_{r\chi}N  - \mathcal{K}^2(1-\chi^2)\,\partial_r\omega\,\partial_\chi\omega \right] \\ && - N^2\,\partial_r \mathcal{K} \left[ (1-\chi^2)\partial_\chi \mathcal{K} + 2\mathcal{K}\chi \right].
\end{eqnarray}
Moreover, equations \eqref{Gtphi} and \eqref{Gphiphi} form an overdetermined linear algebraic system for the single unknown $P_\varphi$. Requiring this system to be consistent leads to the condition $G_{t\varphi} + \omega G_{\varphi\varphi}=0$ which can be rewritten in the form
\begin{equation}\label{cond2}
S_1(N,\mathcal{K})\partial_{rr}\omega + S_2(N,\mathcal{K},\Sigma)\partial_{\chi\chi}\omega + S_3(N,\mathcal{K})\partial_r\omega + S_4(N,\mathcal{K})\partial_{\chi}\omega = 0
\end{equation}
with
\begin{eqnarray}
S_1(N,\mathcal{K}) &=& N \mathcal{K} \Delta (1-\chi^2),\\
S_2(N,\mathcal{K}) &=& N \mathcal{K} (1-\chi^2)^2, \\
S_3 (N,\mathcal{K}) &=& \frac{1}{2}N \mathcal{K}\Delta (1-\chi^2) \,
\partial_r\ln\left(\frac{\mathcal{K}^3\Delta}{N^2}\right),\\
S_4 (N,\mathcal{K}) &=& N \mathcal{K} (1-\chi^2) \left[ \frac{1}{2}(1-\chi^2)\, \partial_\chi\ln\!\left(\frac{\mathcal{K}^3}{N^2}\right) +4\chi \right].
\end{eqnarray} 
If we set $N=1$ (vanishing redshift), choose $\mathcal{K}(r,\chi) = \mathcal{K}(r)$ and $\omega(r,\chi)=\omega(r)$, then the pair of conditions \eqref{cond1} and \eqref{cond2} reduces to
\begin{eqnarray}
(a^2 + r^2)\, \dot{\mathcal{K}} - 2r \, \mathcal{K}(r)  &=& 0 \label{odeK} \\
\ddot \omega + \frac{1}{2} \left( 3\frac{\dot{\mathcal{K}}}{\mathcal{K}} + \frac{\dot \Delta}{\Delta} \right) \dot \omega &=& 0, \label{odeomega}
\end{eqnarray}
where the dot denotes differentiation with respect to $r$. It is worth emphasising that Eq.~\eqref{odeK} is completely decoupled from the shape function $b(r)$. Consequently, the radial function $\mathcal{K}(r)$ can be determined independently of the particular wormhole geometry under consideration. More explicitly, Eq.~\eqref{odeK} can be rewritten as
\begin{equation}
\frac{\dot{\mathcal{K}}}{\mathcal{K}}=\frac{2r}{r^2+a^2},
\end{equation}
which is immediately integrable and gives
\begin{equation}
    \mathcal{K}(r)=c_1(a)\left(r^2+a^2\right),
\end{equation}
where $c_1(a)$ is an integration constant that may, in principle, depend on the oblateness parameter $a$. This constant is fixed by the asymptotic normalisation of the azimuthal metric component. Indeed, asymptotic flatness requires that, as $r\to\infty$,
\begin{equation}
    g_{\varphi\varphi}
    =
    (1-\chi^2)\mathcal{K}(r)
    \sim
    r^2(1-\chi^2),
\end{equation}
so that $\mathcal{K}(r)\sim r^2$. Since
\begin{equation}
    \mathcal{K}(r)
    =
    c_1(a)\left(r^2+a^2\right)
    \sim
    c_1(a)r^2,
    \qquad r\to\infty,
\end{equation}
the asymptotic condition fixes $c_1(a)=1$. Therefore, we have
\begin{equation}\label{Ksol}
\mathcal{K}(r)=r^2+a^2 .
\end{equation}
Substituting Eq.~\eqref{Ksol} into Eq.~\eqref{odeomega}, one obtains
\begin{equation}
\ddot\omega+ \frac{1}{2} \left[ \frac{6r}{r^2+a^2} + \frac{\dot\Delta}{\Delta} \right]\dot\omega=0.
\label{omegaeqK}
\end{equation}
This equation admits the general solution
\begin{equation}
\omega(r)=c_3+c_2\int^r\frac{d\rho}
{(\rho^2+a^2)^{3/2}\sqrt{\Delta(\rho)}}.
\label{omegaformal}
\end{equation}
Thus, for the Morris--Thorne wormhole $b(r)=r_0^2/r$, the frame-dragging function is given by
\begin{equation}\label{omegaMT}
\omega(r)=c_3+c_2\int^r\frac{d\rho}{(\rho^2+a^2)^{3/2}\sqrt{\rho^2+a^2-r_0^2}}.
\end{equation}
Let 
\begin{equation}
k=\frac{a}{r_0},\qquad
U(r):=\frac{r_0}{\sqrt{r^2+a^2}}.
\end{equation}
After imposing the asymptotic condition $\omega(r)\sim 2J_{\rm ADM}/r^3$, Eq.~\eqref{omegaMT} may be written directly as
\begin{equation}
\omega(r)=6J_{\rm ADM}\int_r^\infty\frac{d\rho}{(\rho^2+a^2)^{3/2}\sqrt{\rho^2+a^2-r_0^2}}. 
\label{omegaintegral}
\end{equation}
If we use the substitution $u=r_0/\sqrt{\rho^2+a^2}$, this gives
\begin{equation}
\omega(r)=\frac{6J_{\rm ADM}}{r_0^3}  \int_0^{U(r)}\frac{u^2 du}{\sqrt{(1-u^2)(1-k^2 u^2)}}.  
\end{equation}
The integral above can be easily evaluated using the standard elliptic-integral identity, and we obtain
\begin{equation}\label{omegaexact}
\omega(r)=\frac{6J_{\rm ADM}}{r_0 a^2}\left[F\left(\frac{r_0}{\sqrt{r^2+a^2}},\frac{a}{r_0}\right)-E\left(\frac{r_0}{\sqrt{r^2+a^2}},\frac{a}{r_0}\right)\right],    
\end{equation}
where $F$ and $E$ denote the incomplete elliptic integrals of the first and second kind, respectively, in the convention of Ref.~\cite{Abramowitz1972}. They are defined as
\begin{equation}
F(z,k)=\int_0^z\frac{d\tau}{\sqrt{(1-\tau^2)(1-k^2\tau^2)}},\quad
E(z,k)=\int_0^z\sqrt{\frac{1-k^2\tau^2}{1-\tau^2}}~d\tau.
\end{equation}

As a further consistency check, we compute the Komar angular momentum associated with the axial Killing vector field. Details are given in Appendix~\ref{app:komar}. The corresponding surface integral yields precisely $J_{\rm ADM}$, confirming that it is the ADM/Komar angular momentum measured at the asymptotic end used to normalise the dragging profile. Since the shape function has been fixed to $b(r)=r_0^2/r$, the metric contains no $1/r$ correction in the asymptotic expansion of $g_{rr}$. The ADM mass therefore vanishes, i.e. $M_{\rm ADM}=0$. Consequently, the usual Kerr relation \eqref{Kerr-rel} is not applicable in the present vanishing-ADM-mass asymptotic sector.
It is important that the asymptotic normalisation of $\omega$ has already been fixed by the coefficient $J_{\rm ADM}$ in Eq.~\eqref{omegaintegral}. No further relation involving $a$ is needed for that purpose. After the field equations, asymptotic flatness, and the Komar normalisation have been imposed, no equation remains that ties $J_{\rm ADM}$ to the Kerr-like parameter $a$. At fixed throat scale $r_0$, the exact rotating family is therefore two-dimensional. We define
\begin{equation}\label{independentDimensionlessParameters}
\hat a:=\frac{a}{r_0},\qquad
\hat J:=\frac{J_{\rm ADM}}{r_0^2}.
\end{equation}
The regularity condition constrains only the oblate parameter, $|\hat a|<1$, while $\hat J$ is an independent finite angular-momentum parameter. Because the metric functions $\Sigma$, $\Delta$, and $\mathcal K$ contain $a$ only through $a^2$, the sign of the physical rotation is, in the unrestricted family, the sign of $J_{\rm ADM}$. For the sole purpose of displaying a one-parameter sequence in the numerical ergoregion and equatorial-capture plots, we select the phenomenological slice
\begin{equation}\label{JADMParam}
\hat J=\hat a,\qquad\text{equivalently}\qquad
J_{\rm ADM}=a r_0.
\end{equation}
Equation~\eqref{JADMParam} is neither a Kerr identity nor a consequence of the Einstein equations. It is an optional closure that correlates the otherwise independent oblateness and dragging amplitudes. Other curves $J_{\rm ADM}=J_{\rm ADM}(a)$ are equally compatible with the exact solution. The correct static limit also concerns both parameters. The metric reduces to the zero-redshift Morris--Thorne seed when $a\to0$ and $J_{\rm ADM}\to0$. To see this directly, Eq.~\eqref{omegaintegral} gives the small-$a$ expansion
\begin{equation}
\frac{1}{(\rho^2+a^2)^{3/2}\sqrt{\rho^2+a^2-r_0^2}}=\frac{1}{\rho^3\sqrt{\rho^2-r_0^2}}+\mathcal{O}(a^2).
\end{equation}
Consequently,
\begin{align}
6J_{\rm ADM}\int_r^\infty\frac{d\rho}{(\rho^2+a^2)^{3/2}\sqrt{\rho^2+a^2-r_0^2}}
&=6J_{\rm ADM}\left[\frac{1}{2r_0^3}\arctan\left(\frac{r_0}{\sqrt{r^2-r_0^2}}\right)-\frac{\sqrt{r^2-r_0^2}}{2r_0^2 r^2}\right]+\mathcal{O}(J_{\rm ADM}a^2).
\end{align}
Hence $\omega(r)=\mathcal O(J_{\rm ADM})+\mathcal O(J_{\rm ADM}a^2)$ near $a=0$, and every path for which $J_{\rm ADM}\to0$ yields the static limit. On the illustrative slice \eqref{JADMParam}, $\omega=\mathcal O(a)$ and the next correction is $\mathcal O(a^3)$. By contrast, taking $a\to0$ at fixed nonzero $J_{\rm ADM}$ does not produce a static spacetime. It leaves a nonzero dragging profile. Conversely, taking $J_{\rm ADM}\to0$ at fixed $a\neq0$ yields a static but oblate member of the family, not the spherical Morris--Thorne seed.
\begin{table*}[t]
\centering
\caption{Logical dependence of the principal results on the optional prescription $J_{\rm ADM}=ar_0$. The exact field-equation solution is the family $(r_0,a,J_{\rm ADM})$. The last column identifies only the results specialised to the one-parameter slice $\hat J=\hat a$.}
\label{tab:parameter-dependence}
\renewcommand{\arraystretch}{1.2}
\begin{tabular}{p{0.30\textwidth}p{0.45\textwidth}p{0.17\textwidth}}
\toprule
Result & Parameter dependence in the exact family & Uses $\hat J=\hat a$?\\
\midrule
$\mathcal K=r^2+a^2$ and the elliptic-integral profile $\omega$ & Derived from the field equations; $\omega$ is linear in the independently normalised $J_{\rm ADM}$ & No\\
Quasi-local throat geometry and regularity range & The minimal-surface and null-expansion conditions depend on $r_0$ and $a$; the coordinate representative is $r_{\rm th}=\sqrt{r_0^2-a^2}$, and the geometry is regular for finite $J_{\rm ADM}$ when $|a|<r_0$ & No\\
Two-ended extension, temporal function, and timelike throat & Hold for the exact profile with any finite $J_{\rm ADM}$ & No\\
Stress tensor and energy conditions & Depend separately on $\hat a$ and $\hat J$; the radial NEC and equatorial WEC violations do not require the prescription & No\\
Numerical critical value, equatorial-capture Table~\ref{tab:shadow-observables}, and Figs.~\ref{fig:ergoregion}, \ref{fig:shadow-match-horizon}, and \ref{fig:shadow-match-static} & General thresholds and equatorial capture endpoints are two-parameter; the displayed numbers follow a one-parameter path & Yes\\
\bottomrule
\end{tabular}
\end{table*}

The resulting exact line element of the two-parameter rotating family is
\begin{align}\label{exactrotmet}
ds^2 = & -[1-(1-\chi^2) (r^2+a^2 ) \omega^2]dt^2+\frac{r^2+a^2\chi^2}{r^2 + a^2 - r b(r)}dr^2-2(1-\chi^2) (r^2+a^2 )\omega dtd\varphi \notag \\
       &+\frac{(r^2+a^2\chi^2) }{1-\chi^2} d\chi^2+(1-\chi^2) (r^2+a^2 ) d\varphi^2,
\end{align}
with $\omega(r)$ given by \eqref{omegaexact} and $b(r)=r_0^2/r$. The scalar-polynomial regularity of the solution for $0\leq |a|<r_0$ is checked explicitly in Appendix~\ref{app:regularity}. The limiting value $|a|=r_0$ is excluded from the regular family.

\subsection{Source components and local energy conditions}

The exact solution determines not only the metric functions but also the anisotropic source supporting the geometry. We now give the tetrad components of this source for the final metric \eqref{exactrotmet}. This is useful both as a check of the static limit and as a first diagnostic of the energy conditions. Since the matter tensor is diagonal in the orthonormal frame $\{U^\alpha,n_{(r)}^\alpha,n_{(\chi)}^\alpha,n_{(\varphi)}^\alpha\}$, the relevant projections are
\begin{equation}
8\pi \widetilde\rho=G_{\alpha\beta}U^\alpha U^\beta,\qquad
8\pi P_i=G_{\alpha\beta}n_{(i)}^\alpha n_{(i)}^\beta,\qquad
i=r,\chi,\varphi.
\end{equation}
For compactness, let $q=1-\chi^2$, $\Sigma=r^2+a^2\chi^2$, $\mathcal K=r^2+a^2$, and $\Delta=r^2+a^2-r_0^2$. Using the radial equation for the frame-dragging function,
\begin{equation}
\dot\omega(r)=-\frac{6J_{\rm ADM}}{\mathcal K^{3/2}\sqrt{\Delta}},
\end{equation}
we define the nonnegative rotation contribution
\begin{equation}\label{Wdef}
\mathcal W(r,\chi):=\frac{q\,\mathcal K\,\Delta}{4\Sigma}\dot\omega^2=
\frac{9qJ_{\rm ADM}^2}{\mathcal K^2\Sigma}.
\end{equation}
The apparent singularity of $\dot\omega$ at the throat is therefore absent from the tetrad components, because only the combination $\Delta\dot\omega^2$ enters. A direct projection of the Einstein tensor gives
\begin{subequations}\label{sourceComponents}
\begin{align}
8\pi \widetilde\rho&=\frac{r_0^2\,\mathcal N_\rho}{\mathcal K^2\Sigma^3}
-\mathcal W,\label{rhoSource}\\
8\pi P_r&=-\frac{r^2r_0^2}{\mathcal K\Sigma^2}+\mathcal W,\label{PrSource}\\
8\pi P_\chi&=\frac{r_0^2\left(r^4-a^4\chi^2\right)}{\mathcal K^2\Sigma^2}-\mathcal W,\label{PchiSource}\\
8\pi P_\varphi&=\frac{r_0^2\left(r^2-a^2\chi^2\right)}{\Sigma^3}-3\mathcal W,\label{PphiSource}
\end{align}
\end{subequations}
where
\begin{equation}\label{NrhoDef}
\mathcal N_\rho=-r^6-a^2(1-\chi^2)r^4+a^4(4\chi^2-1)r^2+a^6\chi^2(1+\chi^2).
\end{equation}
As a consistency check, in the joint static limit $(a,J_{\rm ADM})\to(0,0)$, these expressions reduce to
\begin{equation}
8\pi\widetilde\rho=8\pi P_r=-\frac{r_0^2}{r^4},\qquad
8\pi P_\chi=8\pi P_\varphi=\frac{r_0^2}{r^4},
\end{equation}
which is precisely the zero-redshift Morris--Thorne result for $b(r)=r_0^2/r$. For a diagonal stress-energy tensor in an orthonormal frame, the null energy condition (NEC) is equivalent to
\begin{equation}
\widetilde\rho+P_r\ge0,\qquad
\widetilde\rho+P_\chi\ge0,\qquad
\widetilde\rho+P_\varphi\ge0 ,
\end{equation}
while the weak energy condition (WEC) additionally requires $\widetilde\rho\ge 0$. These conditions are central in the usual Morris--Thorne flare-out analysis and in assessments of exotic matter in general relativity \cite{Morris1988AJP, Visser1996, Lobo2008CQGR, Wall2013CQG}. For the general family, the radial NEC combination is especially simple because the rotational contribution $\mathcal W$ cancels between $\widetilde\rho$ and $P_r$. Using the independent parameters \eqref{independentDimensionlessParameters}, the throat value is
\begin{equation}\label{radialNECthroat}
8\pi\left(\widetilde\rho+P_r\right)_{\rm th}=
\frac{\hat a^6q^2-3\hat a^4q+4\hat a^2-2}
{r_0^2\left(1-\hat a^2q\right)^3}.
\end{equation}
This expression contains no $\hat J$. In particular, on the equatorial section $q=1$,
\begin{equation}\label{radialNECeq}
8\pi\left(\widetilde\rho+P_r\right)_{\rm th,eq}
=-\frac{\hat a^4-2\hat a^2+2}{r_0^2(1-\hat a^2)^2}<0,
\qquad |\hat a|<1.
\end{equation}
Thus, the equatorial radial NEC violation is independent of the optional relation between $a$ and $J_{\rm ADM}$. The equatorial energy density is
\begin{equation}\label{rhoEqThroat}
8\pi\widetilde\rho_{\rm th,eq}
=-\frac{\hat a^4-\hat a^2+1+9\hat J^2(1-\hat a^2)}
{r_0^2(1-\hat a^2)^2}<0,
\end{equation}
so the WEC is violated for every regular $\hat a$ and every finite $\hat J$. The angular NEC combinations are
\begin{equation}
8\pi\left(\widetilde\rho+P_\chi\right)_{\rm th,eq}
=-\frac{\hat a^2+18\hat J^2(1-\hat a^2)}
{r_0^2(1-\hat a^2)^2}\leq0,
\end{equation}
and
\begin{equation}
8\pi\left(\widetilde\rho+P_\varphi\right)_{\rm th,eq}
=\frac{\hat a^2-36\hat J^2}{r_0^2(1-\hat a^2)}.
\end{equation}
The polar-angular combination is strictly negative away from the joint static point. The azimuthal combination can have either sign in the unrestricted two-parameter family. On the illustrative slice $\hat J=\hat a\equiv j$, it reduces to $-35j^2/[r_0^2(1-j^2)]<0$. The robust conclusion needed for the wormhole support is already supplied by Eqs.~\eqref{radialNECeq} and \eqref{rhoEqThroat}: the physical stress tensor violates the radial NEC and the WEC at the equatorial throat independently of the phenomenological prescription. Rotation redistributes the remaining components angularly. For example, the radial NEC at the poles is $8\pi(\widetilde\rho+P_r)_{\rm th,pole}=(4\hat a^2-2)/r_0^2$, which changes sign for sufficiently large oblateness. The asymptotic behaviour leads to the same conclusion. As $r\to\infty$,
\begin{eqnarray}
8\pi\widetilde\rho&=&-\frac{r_0^2}{r^4}+\mathcal O(r^{-6}),\qquad
8\pi P_r=-\frac{r_0^2}{r^4}+\mathcal O(r^{-6}),\\
8\pi P_\chi&=&\frac{r_0^2}{r^4}+\mathcal O(r^{-6}),\qquad
8\pi P_\varphi=\frac{r_0^2}{r^4}+\mathcal O(r^{-6}).
\end{eqnarray}
Therefore,
\begin{equation}
  8\pi\left(\widetilde\rho+P_r\right)=-\frac{2r_0^2}{r^4}+\mathcal O(r^{-6})<0
\end{equation}
for sufficiently large $r$. The leading asymptotic NEC and WEC violations are inherited from the Morris--Thorne shape function $b(r)=r_0^2/r$, while the rotational corrections enter at subleading order. This behaviour is consistent with the standard flare-out expectation for traversable wormholes in Einstein gravity and with the stress-energy diagnostics performed for other exact rotating wormhole constructions \cite{Teo1998PRD, Batic2026EPJC}.

\subsubsection{Physical status of the supporting source}
\label{subsec:source-status}
The tensor in Eqs.~\eqref{sourceComponents} is mathematically well defined and covariantly conserved, because it is $G_{\mu\nu}/(8\pi)$ and the contracted Bianchi identity gives $\nabla_\mu T^{\mu\nu}=0$. Its diagonal form in the ZAMO tetrad permits an anisotropic-fluid decomposition, but this decomposition is algebraic because it does not provide a matter action, independent matter field equations, a particle current, a constitutive relation, or an equation of state. In particular, the functions $\widetilde\rho$, $P_r$, $P_\chi$, and $P_\varphi$ depend on both $r$ and $\chi$, and no barotropic closure has been derived. Therefore, quantities such as sound speeds, microscopic causality, and material stability cannot be inferred from the reconstructed tensor alone. At the joint static point $a=J_{\rm ADM}=0$, the source does admit the familiar Ellis--Bronnikov phantom-scalar realisation \cite{Ellis1973JMP, Bronnikov1973APP}. In the proper radial coordinate $r^2=\ell^2+r_0^2$, the reversed-sign massless-scalar action and profile
\begin{equation}\label{staticPhantomRealisation}
S_{\phi}=\frac{1}{2}\int d^4x\,\sqrt{-g}\,
g^{\mu\nu}\partial_\mu\phi\,\partial_\nu\phi,\qquad
\phi(\ell)=\frac{1}{\sqrt{4\pi}}\arctan\!\left(\frac{\ell}{r_0}\right)
\end{equation}
give $8\pi\widetilde\rho=8\pi P_r=-r_0^2/r^4$ and $8\pi P_\chi=8\pi P_\varphi=r_0^2/r^4$, exactly reproducing the static limit above. This is an explicit field-theoretic completion of the seed, but the reversed kinetic sign is ghost-like and the source violates the NEC. It is therefore not conventional stable matter. This static scalar realisation does not extend automatically to the general rotating family. In particular, the reconstructed source generically has $P_\chi\neq P_\varphi$ and distinct $J_{\rm ADM}^2$ contributions in the angular principal stresses, whereas the static single-scalar ansatz produces equal angular pressures. Additional fields, couplings, or effective gravitational terms would be required, and no such completion is derived here. Specific rotating wormholes supported by phantom scalars, electromagnetic fields, three-form sectors, or effective modified-gravity terms are known examples in which a matter action is supplied from the outset \cite{Kleihaus2014PRD, Chew2019PRD, Clement2023PLB, Tanghpati2024NPB, Antoniou2020PRD}. However, the existence of those models does not establish that the present tensor can be generated by any one of them. The strict radial NEC violation at the equatorial throat, together with the negative energy density there, excludes support solely by any matter sector whose total stress tensor obeys the NEC and WEC in Einstein gravity. The source is also not compactly supported because its leading tetrad components decay as $r^{-4}$ and extend to both asymptotic regions. Accordingly, in this paper the rotating source is regarded as an effective phenomenological anisotropic stress tensor required by the exact geometry, not as an established realistic material model. A microphysical completion would require finding fields $\Psi$ and an action for which $T_{\mu\nu}[\Psi]=G_{\mu\nu}/(8\pi)$, verifying the independent equations of motion for $\Psi$, and analysing perturbative stability. We make no such claim here.

\section{Main Results}

Having derived the exact rotating geometry with independent $a$ and $J_{\rm ADM}$, we now discuss its geometrical and optical properties. We first use the original radial chart and embedding diagrams only as a visualisation, and then identify the throat quasi-locally as a compact minimal surface whose orthogonal null expansions vanish. The regular extension, absence of closed timelike curves, and timelike character of the throat are established for the full family $(r_0,a,J_{\rm ADM})$. The ergoregion and equatorial optical formulae are likewise written first in terms of the independent variables $\hat a$ and $\hat J$. Only the numerical critical value and the plotted/tabulated equatorial capture intervals are subsequently restricted to the illustrative slice $\hat J=\hat a\equiv j$. No complete two-dimensional shadow is inferred.

\subsection{Coordinate description and oblateness bound}
This subsection provides a useful chart-based visualisation and the reality bound on the oblateness parameter. The throat itself is defined geometrically, rather than by a coordinate value, in Sec.~\ref{subsec:minimal-area-throat}. 
Let us now examine the intrinsic geometry of the spatial sections associated with the rotating wormhole spacetime. To this end, we restrict attention to a constant-time hypersurface and, within it, select a two-dimensional constant-latitude section specified by
\begin{equation}
    \chi=\chi_0, \qquad 0 \leq |\chi_0|<1 .
\end{equation}
The induced two-dimensional line element on this surface is then given by
\begin{equation}\label{equatorialSlice}
    ds^2
    =
    \frac{r^2+a^2\chi_0^2}{\Delta(r)}\,dr^2
    +(1-\chi_0^2)\left(r^2+a^2\right)d\varphi^2 .
\end{equation}
It is useful to introduce the equatorial circumferential radius $\rho=\sqrt{r^2+a^2}$. This quantity has a coordinate-independent meaning on the equatorial fixed set because if $\eta^\mu=(\partial_\varphi)^\mu$ is the axial Killing field normalised to have $2\pi$-periodic orbits, then the proper circumference of an equatorial orbit is $C_{\rm eq}=2\pi\sqrt{\eta^\mu\eta_\mu}=2\pi\rho$. Since $2\rho\,d\rho=2r\,dr$, $r^2\,dr^2=\rho^2\,d\rho^2$, and $\Delta(r)=r^2+a^2-r_0^2=\rho^2-r_0^2$, the metric \eqref{equatorialSlice} can be written as
\begin{equation}\label{MTSlice}
ds^2=\frac{\rho^2 \left[ \rho^2 - a^2  (1-\chi_0^2)  \right]}{(\rho^2-a^2) (\rho^2-r_0^2)} d\rho^2+ (1-\chi_0^2 ) \rho^2 d\varphi^2 .
\end{equation}
This form is particularly useful because it makes explicit the role of \(\rho\) as the circumferential radius, while also separating the coordinate singular behaviour associated with \(\Delta(r)=0\) from the angular rescaling induced by the fixed value of \(\chi_0\). Taking $\chi_0 = 0$, we precisely obtain the standard spatial geometry of the static Morris--Thorne wormhole, expressed in terms of the equatorial circumferential radius $\rho$.
\begin{equation}\label{MTSliceEquator}
  ds^2=\frac{ d\rho^2  }{1- \dfrac{r_0^2}{\rho^2}} + \rho^2 d\varphi^2 .
\end{equation}
The minimum $C_{\rm eq}=2\pi r_0$ identifies the equatorial section of the candidate neck without referring to a radial coordinate. In the present chart, this same surface is represented by $\rho=r_0$, or equivalently by $\Delta(r)=0$. Solving the latter coordinate equation gives
\begin{equation}\label{rthA}
  r_{\rm th}(a)=\sqrt{r_0^2-a^2}.
\end{equation}
Equation~\eqref{rthA} concerns only the coordinate representative of the neck. Its value decreases monotonically as the oblateness magnitude $|a|$ is increased, while the invariant equatorial circumference remains $2\pi r_0$. The analogy with the Boyer--Lindquist location of the Kerr horizon is therefore purely coordinate-based and should not be confused with the invariant throat definition or with a bound on the independent ADM angular momentum. Reality of this chart requires $|a|\leq r_0$, and the curvature analysis in Appendix~\ref{app:regularity} sharpens this to the regular range $0\leq|a|<r_0$. The endpoint $|a|=r_0$ is singular because $\Sigma_{\rm th}$ vanishes at the equator. For a generic constant-latitude section with \(\chi_0\neq 0\), the proper circumferential radius of the corresponding latitude circle is not \(\rho\) itself, but rather
\begin{equation}
    R(r,\chi_0)
    =
    \sqrt{(1-\chi_0^2)(r^2+a^2)}
    =
    \sqrt{1-\chi_0^2}\,\rho .
\end{equation}
In particular, the minimum value of this radius occurs at the throat and is given by
\begin{equation}
    R_{\rm th}(\chi_0)
    =
    \sqrt{1-\chi_0^2}\,r_0 .
\end{equation}
Thus, away from the equatorial plane, the latitude circles have a smaller proper circumference, as expected from the angular factor \(1-\chi_0^2\). To visualise the intrinsic geometry of the fixed-\(\chi_0\) surface, one may embed it as a surface of revolution in three-dimensional Euclidean space. The Euclidean line element for such an embedding can be written as
\begin{equation}
    ds_{\mathbb{E}^3}^2
    =
    \left[1+\left(\frac{dz}{dR}\right)^2\right]dR^2
    +
    R^2 d\varphi^2 .
\end{equation}
Matching the angular part with the intrinsic metric fixes \(R=R(r,\chi_0)\), while matching the radial part yields
\begin{equation}
    \left(\frac{dz}{dR}\right)^2
    =
    \frac{(r^2+a^2\chi_0^2)(r^2+a^2)}
    {(1-\chi_0^2)r^2\Delta(r)}
    -1 .
\end{equation}
This expression displays the characteristic flaring behaviour of the wormhole throat. Indeed, as the throat is approached one has $\Delta(r)\longrightarrow 0$, and therefore
\begin{equation}
    \left|\frac{dz}{dR}\right|\longrightarrow \infty .
\end{equation}
Hence, the embedded surface becomes vertical at the candidate neck, which is the standard visual signature of a Morris--Thorne-type wormhole. The divergence of \(dz/dR\) should not be interpreted as a curvature singularity of the regular branch \(0\leq |a|<r_0\); rather, it reflects the fact that the circumferential radius reaches a minimum there, so that the embedding surface flares outward on both sides. An embedding diagram depends on the chosen spatial slice and coordinates and is therefore not used as the definition of the throat. The quasi-local characterisation follows next.

\begin{figure}[H]
    \centering
    \includegraphics[width=0.5\linewidth]{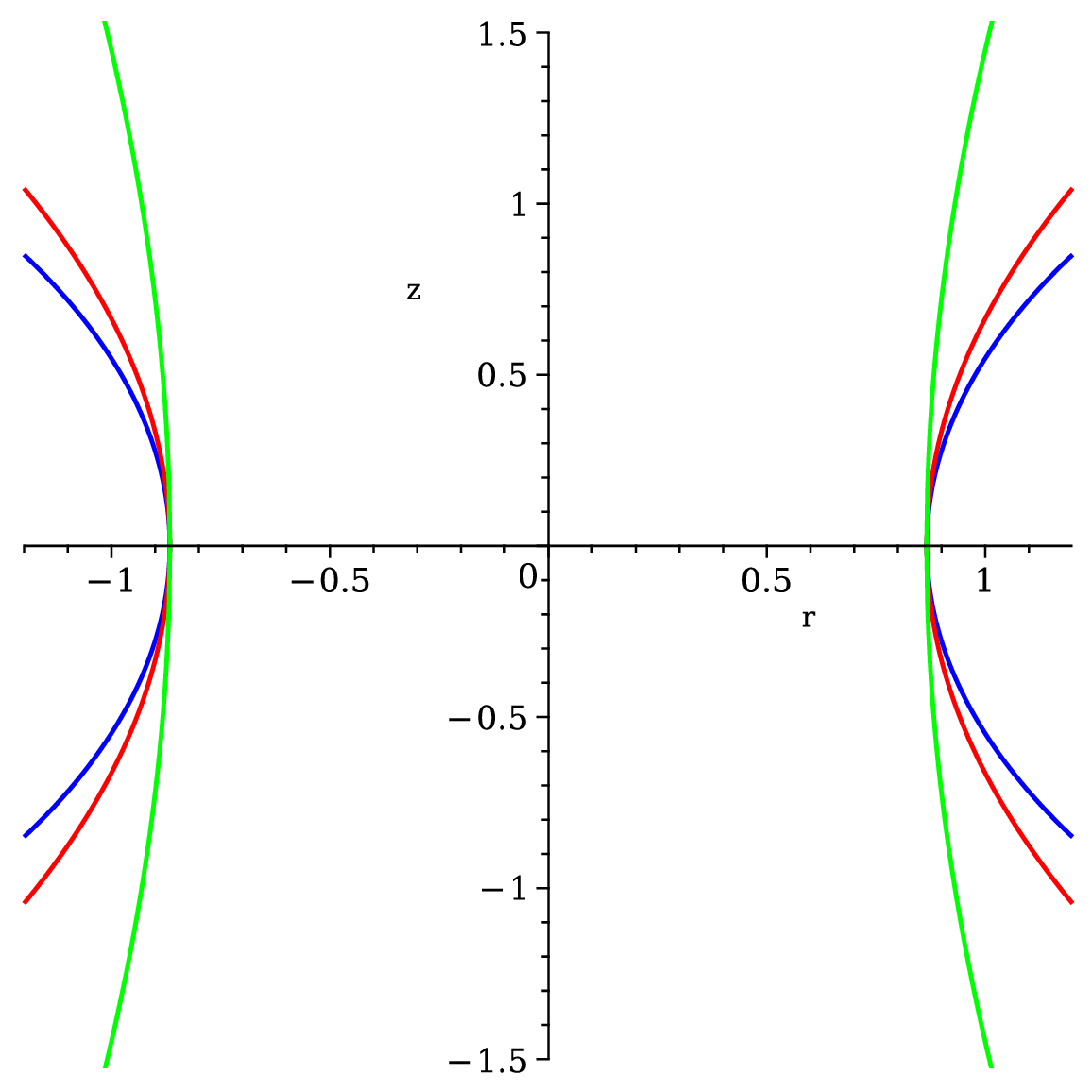}
    \caption{Coordinate embedding profiles for $r_0=1$ and $a=0.5$ at several fixed values of $\chi$. The solid blue curve corresponds to $\chi=0$, the solid red curve to $\chi=0.5$, and the solid green curve to $\chi=0.9$. These profiles visualise flare-out but do not define the throat.}
    \label{fig:ThroatEmbedding}
\end{figure}

\subsection{Quasi-local throat characterisation: minimal area and null expansions}\label{subsec:minimal-area-throat}
A stationary rotating spacetime does not, in general, supply a foliation-independent minimal-surface throat. In the present solution, the unit lapse and the time coordinate normalised at the observer's asymptotic end provide a natural, explicitly specified stationary foliation by the hypersurfaces $t=\mathrm{const.}$ used throughout. Let $S_\ell$ denote their closed axial two-surfaces. Within this foliation we define the throat quasi-locally as the unique member $S_0$ with vanishing mean curvature and positive flare-out, following the geometric minimal-surface viewpoint of Refs.~\cite{HochbergVisser1997PRD,HochbergVisser1998PRD,TomikawaIzumiShiromizu2015PRD}. This replaces a coordinate definition by scalar geometrical conditions and is invariant under reparametrisations of the radial coordinate and coordinate changes intrinsic to $S_\ell$. The equations $\ell=0$, $\Delta=0$, and $r=\sqrt{r_0^2-a^2}$ merely locate the geometrically defined surface in the coordinates used here. We compute the mean curvature of the closed two-surfaces $\ell=\mathrm{const}$, where $r^2(\ell)=\ell^2+r_0^2-a^2$. On a constant-time hypersurface, the spatial metric induced by
\eqref{properMetric} is
\begin{equation}\label{spatial-metric-ell}
h_{ij}dx^i dx^j=\frac{\Sigma_\ell}{r^2(\ell)}\,d\ell^2+\frac{\Sigma_\ell}{1-\chi^2}\,d\chi^2+(1-\chi^2)(\ell^2+r_0^2)\,d\varphi^2,
\end{equation}
where $\Sigma_\ell=\ell^2+r_0^2-a^2(1-\chi^2)$, and $r^2(\ell)=\ell^2+r_0^2-a^2$. For $0\leq |a|<r_0$, both $r^2(\ell)$ and $\Sigma_\ell$ are strictly positive at and near the throat, so this spatial metric is regular. The induced metric on the closed two-surface $S_\ell$, defined by $\ell = \mathrm{const.}$, is
\begin{equation}\label{induced-metric-ell}
\sigma_{AB}dx^A dx^B=\frac{\Sigma_\ell}{1-\chi^2}\,d\chi^2+(1-\chi^2)(\ell^2+r_0^2)\,d\varphi^2 ,
\end{equation}
with area element
\begin{equation}\label{area-element-ell}
dA=\sqrt{\det\sigma}\,d\chi d\varphi
=\sqrt{\Sigma_\ell(\ell^2+r_0^2)}\,d\chi d\varphi.
\end{equation}
The outward-pointing unit normal to \(S_\ell\) inside the spatial slice is
\begin{equation}\label{unit-normal-ell}
s^i\partial_i=\frac{r(\ell)}{\sqrt{\Sigma_\ell}}\,\partial_\ell .
\end{equation}
Since the spatial metric is block diagonal with respect to \(\ell\) and the angular coordinates, the mean curvature of \(S_\ell\) is
\begin{equation}\label{mean-curvature-general}
H=\sigma^{AB}D_A s_B=\frac{1}{2\sqrt{h_{\ell\ell}}}
\sigma^{AB}\partial_\ell\sigma_{AB},
\end{equation}
where $D$ is the covariant derivative associated with $h_{ij}$. Using
\begin{equation}
h_{\ell\ell}=\frac{\Sigma_\ell}{r^2(\ell)},\qquad
\sigma_{\chi\chi}=\frac{\Sigma_\ell}{1-\chi^2},\qquad
\sigma_{\varphi\varphi}=(1-\chi^2)(\ell^2+r_0^2),
\end{equation}
one obtains
\begin{equation}
\sigma^{AB}\partial_\ell\sigma_{AB}=2\ell\left(\frac{1}{\Sigma_\ell}+\frac{1}{\ell^2+r_0^2}\right).
\end{equation}
Therefore, the mean curvature is
\begin{equation}\label{mean-curvature-ell}
H(\ell,\chi)=\ell\,\frac{r(\ell)}{\sqrt{\Sigma_\ell}}\left(\frac{1}{\Sigma_\ell}
+\frac{1}{\ell^2+r_0^2}\right).
\end{equation}
This expression has two immediate consequences. First,
\begin{equation}\label{mean-curvature-zero-throat}
H(0,\chi)=0,
\end{equation}
for every \(\chi\in[-1,1]\). Thus, the throat \(S_0\) is a minimal surface in the sense that its mean curvature vanishes. Second, for \(0\leq |a|<r_0\), all factors multiplying \(\ell\) in \eqref{mean-curvature-ell} are strictly positive. Hence
\begin{equation}
	H(\ell,\chi)>0 \quad \text{for} \quad \ell>0,
	\qquad
	H(\ell,\chi)<0 \quad \text{for} \quad \ell<0 .
	\label{mean-curvature-sign}
\end{equation}
The mean curvature, therefore, changes sign at \(\ell=0\), with the two sides of the wormhole expanding away from the throat. This is the standard local geometrical signature of a wormhole throat.

The same result can be phrased directly in terms of the first variation of area. The area of \(S_\ell\) is
\begin{equation}
	A(\ell)
	=
	2\pi
	\int_{-1}^{1}
	\sqrt{
		\Sigma_\ell(\ell^2+r_0^2)
	}\,
	d\chi .
	\label{area-ell-surfaces-mean}
\end{equation}
Using the standard first-variation formula,
\begin{equation}
	\frac{dA}{d\ell}
	=
	\int_{S_\ell}
	\sqrt{h_{\ell\ell}}\,H\,dA ,
	\label{first-variation-area}
\end{equation}
and substituting \eqref{mean-curvature-ell}, one finds
\begin{equation}
	\frac{dA}{d\ell}
	=
	2\pi
	\int_{-1}^{1}
	\ell
	\left(
	\frac{1}{\Sigma_\ell}
	+
	\frac{1}{\ell^2+r_0^2}
	\right)
	\sqrt{\Sigma_\ell(\ell^2+r_0^2)}
	\,d\chi .
	\label{area-derivative-ell}
\end{equation}
The integrand is proportional to \(\ell\), while the remaining factor is strictly positive in the regular range \(0\leq |a|<r_0\). Consequently,
\begin{equation}
	\left.\frac{dA}{d\ell}\right|_{\ell=0}=0,
	\qquad
	\frac{dA}{d\ell}>0 \ \text{for}\ \ell>0,
	\qquad
	\frac{dA}{d\ell}<0 \ \text{for}\ \ell<0 .
\end{equation}
Thus, the area decreases as one approaches \(\ell=0\) from the left and increases as one moves away from \(\ell=0\) to the right. Hence \(S_0\) is the unique minimal-area member of the foliation \(S_\ell\).

Equivalently, the second derivative of the area at the throat is positive:
\begin{equation}
	\left.
	\frac{d^2A}{d\ell^2}
	\right|_{\ell=0}
	=
	2\pi
	\int_{-1}^{1}
	\left(
	\frac{1}{\Sigma_0}
	+
	\frac{1}{r_0^2}
	\right)
	\sqrt{\Sigma_0 r_0^2}
	\,d\chi
	>0 ,
	\label{area-second-variation-throat}
\end{equation}
where $\Sigma_0=r_0^2-a^2(1-\chi^2)>0$. This confirms that the throat is not merely a stationary member of the foliation, but a strict local minimum of the area. The same surface also admits a spacetime, null-expansion characterisation. Let $n^\mu$ be the future-directed unit normal to the stationary slices and let $s^\mu$ be the outward unit normal \eqref{unit-normal-ell} to $S_\ell$ within a slice. Define the two future null normals $k_\pm^\mu=n^\mu\pm s^\mu$, $k_+\cdot k_-=-2$, and their expansions $\Theta_\pm=\sigma^{\mu\nu}\nabla_\mu k^\pm_\nu$. For the ADM form \eqref{properMetric}, the lapse is unity and the shift has only an azimuthal component, $\beta^\varphi=-\Omega(\ell)$. Because the spatial metric is stationary, the extrinsic curvature of the $t=\mathrm{const.}$ slices has only a mixed $\ell\varphi$-component. In particular, its trace tangent to $S_\ell$ vanishes, that is $\sigma^{AB}K_{AB}=0$. Consequently, up to the interchange of the labels $+$ and $-$ associated with the orientation convention, we have
\begin{equation}\label{nullExpansionsEqualH}
\Theta_+=H,\qquad
\Theta_-=-H.
\end{equation}
Equation~\eqref{mean-curvature-zero-throat} therefore implies
\begin{equation}\label{nullExpansionsZero}
\left.\Theta_+\right|_{S_0}=\left.\Theta_-\right|_{S_0}=0.
\end{equation}
Moreover, the outward normal derivative of the mean curvature is pointwise positive because
\begin{equation}\label{invariantFlareOut}
\left.s^\mu\nabla_\mu H\right|_{S_0}=\frac{r_{\rm th}^2}{\Sigma_0}
\left(\frac{1}{\Sigma_0}+\frac{1}{r_0^2}\right)>0,
\qquad |a|<r_0.
\end{equation}
Since $H$ is stationary and axisymmetric, $n^\mu\nabla_\mu H=0$. With Eqs.~\eqref{nullExpansionsEqualH} and \eqref{invariantFlareOut}, the corresponding null flare-out derivatives satisfy
\begin{equation}\label{nullFlareOut}
\left.k_\pm^\mu\nabla_\mu\Theta_\pm\right|_{S_0}=\left.s^\mu\nabla_\mu H\right|_{S_0}>0.
\end{equation}
Thus, $S_0$ is doubly marginal with respect to the two future null normals and has a strict spatial and null flare-out. The vanishing of $H$ and $\Theta_\pm$, together with the positivity in Eqs.~\eqref{area-second-variation-throat} and \eqref{invariantFlareOut}, is the coordinate-invariant throat statement within the canonical stationary foliation. It is stronger and more informative than quoting the value of $r$ at which $\Delta$ vanishes. We do not claim that this construction furnishes a foliation-independent definition for arbitrary dynamical wormholes. Rather, it identifies the throat of the present stationary geometry through scalar properties of the surface and its null normals. The associated world tube $\mathcal T=\bigcup_t S_0$ is shown in Sec.~\ref{subsec:proper-radial} to be timelike rather than null. Finally, setting $\ell=0$ in \eqref{induced-metric-ell}, the intrinsic metric on the throat is
\begin{equation}\label{throat-induced-metric}
d\sigma_{\rm th}^2=\frac{r_0^2-a^2(1-\chi^2)}{1-\chi^2}\,d\chi^2+r_0^2(1-\chi^2)\,d\varphi^2.
\end{equation}
The corresponding throat area is
\begin{equation}\label{throat-area}
A_{\rm th}=2\pi r_0\int_{-1}^{1}\sqrt{r_0^2-a^2(1-\chi^2)}\,d\chi.
\end{equation}
In the nonrotating limit \(a\to0\), this reduces to $A_{\rm th}\longrightarrow 4\pi r_0^2$, that is the usual area of a spherical Morris--Thorne throat. We have therefore established the throat without using its coordinate location as a definition. It is the unique member $S_0$ of the closed family $S_\ell$ in the natural stationary foliation with $H=0$, $A'(0)=0$, $A''(0)>0$, and $\Theta_+=\Theta_-=0$, while the normal flare-out derivative is positive. The statements $\ell=0$, $\Delta=0$, and $r=r_{\rm th}$ are only coordinate representatives while the minimum equatorial circumference and the intrinsic metric \eqref{throat-induced-metric} provide additional geometrical diagnostics.

\subsection{Proper radial coordinate, causal status, and two-ended extension}
\label{subsec:proper-radial}

The geometrically defined throat $S_0$ is represented in the one-sided chart by $\Delta(r_{\rm th})=0$. The coordinate $r$ used in \eqref{exactrotmet} covers one exterior region of the wormhole and becomes singular at this surface. This is analogous to the usual Morris--Thorne radial coordinate, which reaches a minimum at the throat and must be replaced by a proper radial coordinate in order to display the two-sided geometry \cite{Morris1988AJP, Visser1996, Lobo2008CQGR}. In the present rotating geometry, the coefficient $g_{rr}=\Sigma/\Delta$ depends on the polar coordinate. Thus, no single angle-independent radial coordinate can be the proper distance along every meridian. Nevertheless, there is a natural equatorial proper radial coordinate which also provides a regular two-sided radial coordinate for the full spacetime. To this end, we define
\begin{equation}\label{ellDef}
\ell=\pm\sqrt{r^2+a^2-r_0^2}=\pm\sqrt{r^2-r_{\rm th}^2}
\end{equation}
with $r_{\rm th}^2=r_0^2-a^2$. Equivalently,
\begin{equation}\label{rOfEll}
r(\ell)=\sqrt{\ell^2+r_0^2-a^2},\qquad
\Delta(r(\ell))=\ell^2,\qquad
\mathcal K(\ell)=r^2(\ell)+a^2=\ell^2+r_0^2.
\end{equation}
The two signs in \eqref{ellDef} correspond to the two sides of the wormhole. The original exterior region is $\ell>0$, the throat is at $\ell=0$, and the second exterior region is $\ell<0$. The two asymptotic domains are therefore obtained as $\ell\to+\infty$, and  $\ell\to-\infty$. On the equatorial section $\chi=0$, Eq.~\eqref{MTSlice} becomes
\begin{equation}\label{eqProperSlice}
  ds^2_{\rm eq}=d\ell^2+\left(\ell^2+r_0^2\right)d\varphi^2.
\end{equation}
Thus $\ell$ is exactly the proper radial distance on the equatorial constant time slice, and the equatorial circumferential radius $\rho=\sqrt{\ell^2+r_0^2}$ has a minimum $\rho_{\rm th}=r_0$ at $\ell=0$. For the full four-dimensional metric, it is useful to write $q:=1-\chi^2$ and 
\begin{equation}
\Sigma_\ell:=r^2(\ell)+a^2\chi^2=\ell^2+r_0^2-a^2 q.
\end{equation}
Since $dr=(\ell/r(\ell))\,d\ell$, the radial part of the metric transforms as
\begin{equation}
\frac{\Sigma}{\Delta}\,dr^2=\frac{\Sigma_\ell}{r^2(\ell)}\,d\ell^2 .
\end{equation}
Hence, the two-sided form of the metric is
\begin{equation}\label{properMetric}
ds^2=-dt^2+\frac{\Sigma_\ell}{r^2(\ell)}\,d\ell^2+\frac{\Sigma_\ell}{1-\chi^2}\,d\chi^2
+(1-\chi^2)(\ell^2+r_0^2)\left[d\varphi-\Omega(\ell)dt\right]^2.
\end{equation}
For $0\leq |a|<r_0$, one has
\begin{equation}
r^2(0)=r_0^2-a^2>0,\qquad
\Sigma_\ell\big|_{\ell=0}=r_0^2-a^2(1-\chi^2)\geq r_0^2-a^2>0 .
\end{equation}
Therefore
\begin{equation}
g_{\ell\ell}\big|_{\ell=0}=\frac{r_0^2-a^2(1-\chi^2)}{r_0^2-a^2}
\end{equation}
is finite and positive for the regular family. Moreover,
\begin{equation}
\det g=-\frac{(\ell^2+r_0^2)\Sigma_\ell^2}{r^2(\ell)}
\end{equation}
is finite and nonzero at the throat for $0\leq |a|<r_0$, apart from the standard coordinate degeneracy of the azimuthal angle on the symmetry axis. Thus, the singularity of $g_{rr}$ at $\Delta=0$ is a coordinate singularity of the one-sided $r$-chart. A technical point is important for the rotating case. If one simply takes $\omega(r)$ from \eqref{omegaexact} and substitutes $r=r(|\ell|)$, then $\omega$ is continuous but not differentiable at $\ell=0$, because $d\omega/dr\propto\Delta^{-1/2}$. The smooth two-sided extension is instead obtained by using the signed coordinate $\ell$. We define
\begin{equation}\label{OmegaEll}
\Omega(\ell)=6J_{\rm ADM}\int_{\ell}^{+\infty}\frac{d\lambda}{\left(\lambda^2+r_0^2\right)^{3/2}\sqrt{\lambda^2+r_0^2-a^2}}.
\end{equation}
For $\ell>0$, this expression reduces exactly to the one-sided profile $\omega(r)$ in \eqref{omegaexact}. Indeed, using $r=\sqrt{\ell^2+r_0^2-a^2}$, one recovers
\begin{equation}
\Omega(\ell)=6J_{\rm ADM}\int_{r(\ell)}^{+\infty}\frac{d\rho}
{(\rho^2+a^2)^{3/2}\sqrt{\rho^2+a^2-r_0^2}}=\omega(r(\ell)).
\end{equation}
However, \eqref{OmegaEll} is regular at the throat because
\begin{equation}\label{OmegaPrimeEll}
\frac{d\Omega}{d\ell}=-\frac{6J_{\rm ADM}}{(\ell^2+r_0^2)^{3/2}\sqrt{\ell^2+r_0^2-a^2}},
\end{equation}
which is finite at $\ell=0$ for $0\leq |a|<r_0$. The divergence of $d\omega/dr$ in the one-sided coordinate is therefore only a consequence of using $r$ instead of the throat-adapted coordinate $\ell$. The right asymptotic region is normalised by construction so that
\begin{equation}\label{OmegaAsympPlus}
\Omega(\ell)=\frac{2J_{\rm ADM}}{\ell^3}+\mathcal O(\ell^{-5}),\qquad
\ell\to+\infty.
\end{equation}
Thus $g_{t\varphi}\sim-2J_{\rm ADM}(1-\chi^2)/\ell$, as required in the exterior region used for the ADM normalisation. The ADM angular momentum is therefore read off at the asymptotic end whose angular coordinate is normalized to be nonrotating. In the present convention this is the \(\ell\to+\infty\) end, where \(\Omega(\ell)\to0\) and
\begin{equation}
	g_{t\varphi}
	\sim
	-\frac{2J_{\rm ADM}}{\ell}\,(1-\chi^2).
\end{equation}
Thus \(J_{\rm ADM}\) refers to the angular momentum measured with respect to the inertial frame at the right asymptotic end. In the second asymptotic region, let $R=-\ell$. Then
\begin{equation}\label{OmegaAsympMinus}
\Omega(\ell)=\Omega_- -\frac{2J_{\rm ADM}}{R^3}+\mathcal O(R^{-5}),\qquad R\to+\infty,
\end{equation}
where
\begin{equation}\label{OmegaMinus}
\Omega_-=6J_{\rm ADM}\int_{-\infty}^{+\infty}
\frac{d\lambda}{\left(\lambda^2+r_0^2\right)^{3/2}\sqrt{\lambda^2+r_0^2-a^2}}
=\frac{12J_{\rm ADM}}{r_0a^2}\left[K\!\left(\frac{a}{r_0}\right)-E\!\left(\frac{a}{r_0}\right)\right].
\end{equation}
Here $K$ and $E$ are complete elliptic integrals. The second asymptotic end is therefore asymptotically flat in the rigidly rotating angular coordinate $\varphi_-=\varphi-\Omega_- t$. In this coordinate, the combination appearing in the metric is $d\varphi-\Omega(\ell)dt =d\varphi_- -\left[\Omega(\ell)-\Omega_-\right]dt$, and the effective frame-dragging function tends to zero as $\ell\to-\infty$. Thus, the geometry connects two asymptotically flat regions through the throat at $\ell=0$, but the nonrotating frames at the two infinities differ by the constant angular velocity $\Omega_-$. This is a standard feature that can occur in rotating wormhole geometries, i.e. one can choose a frame nonrotating at either asymptotic end, but not at both ends simultaneously when $J_{\rm ADM}\ne 0$. The form \eqref{properMetric} also permits a direct check of the causal status of the regular family. In a rotating spacetime, the existence of an ergoregion should not be confused with the existence of a horizon, and the absence of closed timelike curves must be verified rather than assumed \cite{Visser1996, Teo1998PRD, Wald1986, Batic2026EPJC}. The inverse metric corresponding to \eqref{properMetric} has
\begin{equation}\label{inverseCausalComponents}
g^{tt}=-1,\qquad
g^{t\varphi}=-\Omega(\ell),\qquad
g^{\varphi\varphi}=\frac{1}{(1-\chi^2)(\ell^2+r_0^2)}-\Omega^2(\ell),
\end{equation}
and
\begin{equation}
g^{\ell\ell}=\frac{r^2(\ell)}{\Sigma_\ell},\qquad
g^{\chi\chi}=\frac{1-\chi^2}{\Sigma_\ell}.
\end{equation}
Hence, we have $g^{\mu\nu}\nabla_\mu t\nabla_\nu t=g^{tt}=-1$. Thus, $t$ is a global temporal function on the regular two-ended extension. After fixing the time orientation so that future-directed causal curves have $dt/d\lambda>0$, no causal curve can close on itself. Therefore, the regular family $0\leq |a|<r_0$ contains no closed timelike curves. This conclusion is stronger than the elementary check of the axial circles, although these are also non-timelike since $g_{\varphi\varphi}=(1-\chi^2)(\ell^2+r_0^2)\geq0$. The throat is not a horizon. The normal to a surface $\ell={\rm const.}$ has squared norm
\begin{equation}\label{ellNormalNorm}
g^{\mu\nu}\partial_\mu\ell\partial_\nu\ell=g^{\ell\ell}=
\frac{r^2(\ell)}{\Sigma_\ell}>0 ,
\end{equation}
so the hypersurface $\ell=0$ is timelike rather than null. At the throat,
\begin{equation}
g^{\ell\ell}\big|_{\ell=0}=\frac{r_0^2-a^2}{r_0^2-a^2(1-\chi^2)}>0,\qquad
0\leq |a|<r_0 .
\end{equation}
Moreover, the lapse in \eqref{properMetric} is identically unity and the $t$--$\varphi$ block has determinant
\begin{equation}
g_{tt}g_{\varphi\varphi}-g_{t\varphi}^2=-(1-\chi^2)(\ell^2+r_0^2),
\end{equation}
which does not vanish away from the usual symmetry axis. The ergosurface, when present, is therefore not a horizon. It is only the surface on which the stationary Killing vector $\partial_t$ becomes null. The two asymptotic regions are causally connected through the throat. To see this explicitly, consider curves with fixed $\chi$ and angular velocity $d\varphi/dt=\Omega(\ell)$. Along these curves,
\begin{equation}
  ds^2=-dt^2+\frac{\Sigma_\ell}{r^2(\ell)}\,d\ell^2.
\end{equation}
Hence, any curve satisfying
\begin{equation}\label{timelikeCrossingCondition}
  \left|\frac{d\ell}{dt}\right|<\frac{r(\ell)}{\sqrt{\Sigma_\ell}}
\end{equation}
is timelike and can be chosen to cross $\ell=0$ smoothly. Since such curves can be extended toward either $\ell\to+\infty$ or $\ell\to-\infty$, there is no horizon barrier separating the throat from the asymptotic regions in the regular extension. Equatorial radial geodesics provide a simple independent check. For $\chi=0$ and $L=0$, the null geodesic equation gives
\begin{equation}\label{nullCrossingGeodesic}
  \left(\frac{d\ell}{d\widetilde\lambda}\right)^2=1,
\end{equation}
while radial timelike geodesics satisfy
\begin{equation}\label{timelikeCrossingGeodesic}
  \left(\frac{d\ell}{d\tau}\right)^2=E^2-1 .
\end{equation}
Thus, radial null geodesics and radial timelike geodesics with $E>1$ cross the throat. In this precise sense, the throat is causally traversable. We do not, however, claim full Morris--Thorne traversability for finite-size observers, since tidal force bounds, acceleration constraints, backreaction, and dynamical stability are not analysed here.

\subsection{Ergoregion}

The independent dimensionless parameters are $\hat a=a/r_0$ and $\hat J=J_{\rm ADM}/r_0^2$, with $|\hat a|<1$. For fixed $\hat a$, the magnitude $|\hat J|$ controls the strength of frame dragging, while the sign of $\hat J$ fixes its orientation. In terms of $x=r/r_0$, the metric component relevant for the ergoregion is
\begin{equation}
g_{tt}(x,\chi)=-\left[1-(1-\chi^2)(x^2+\hat a^2)\widetilde\omega^{\,2}(x)\right],
\end{equation}
where
\begin{equation}\label{omegaDimensionlessGeneral}
\widetilde\omega(x):=r_0\omega(r_0x)=\frac{6\hat J}{\hat a^2}\left[
F\left(\frac{1}{\sqrt{x^2+\hat a^2}},\hat a\right)
-E\left(\frac{1}{\sqrt{x^2+\hat a^2}},\hat a\right)
\right].
\end{equation}
At $\hat a=0$, Eq.~\eqref{omegaDimensionlessGeneral} is understood through the finite limit of the integral representation \eqref{omegaintegral}. Define the nonnegative function $
H(x,\chi):=(1-\chi^2)(x^2+\hat a^2)\widetilde\omega^{\,2}(x)$. An ergoregion exists if and only if $\sup H>1$. The angular factor is maximal at $\chi=0$, so the onset is controlled by
\begin{equation}
H_{\rm eq}(x)=(x^2+\hat a^2)\widetilde\omega^{\,2}(x),\qquad 
x\in[x_{\rm th},\infty),\qquad 
x_{\rm th}=\sqrt{1-\hat a^2}.
\end{equation}
Introduce $u=(x^2+\hat a^2)^{-1/2}$. Then, $u=1$ at the throat and $u\to 0$ at infinity, and
\begin{equation}
H_{\rm eq}=\frac{36\hat J^2}{\hat a^4}\left[\frac{F(u,\hat a)-E(u,\hat a)}{u}\right]^2.
\end{equation}
Using
\begin{equation}
F(u,\hat a)-E(u,\hat a)=\hat a^2\int_0^u
\frac{\tau^2\,d\tau}{\sqrt{(1-\tau^2)(1-\hat a^2\tau^2)}},
\end{equation}
and observing that the integrand is positive and strictly increasing for $0<u<1$ and $0<|\hat a|<1$, one concludes that $[F-E]/u$ increases with $u$. Hence, $H_{\rm eq}$ is maximised at the throat for every fixed pair $(\hat a,\hat J)$. The general onset condition is therefore the critical curve
\begin{equation}\label{criticalErgoCurve}
\frac{6|\hat J|}{\hat a^2}\left[K(|\hat a|)-E(|\hat a|)\right]=1,\qquad 0<|\hat a|<1,
\end{equation}
or equivalently
\begin{equation}
|\hat J|_{\rm c}(\hat a)=\frac{\hat a^2}{6\left[K(|\hat a|)-E(|\hat a|)\right]}.
\end{equation}
The continuous $\hat a\to 0$ limit is $|\hat J|_{\rm c}(0)=2/(3\pi)$. For fixed $\hat a$, no ergoregion is present when $|\hat J|<|\hat J|_{\rm c}(\hat a)$. The ergosurface first touches the throat at equality, and a genuine ergoregion forms for larger $|\hat J|$. On the one-parameter numerical slice $\hat J=\hat a\equiv j$, Eq.~\eqref{criticalErgoCurve} reduces to
\begin{equation}\label{criticalSpinCondition}
\frac{6}{|j|}\left[K(|j|)-E(|j|)\right]=1.
\end{equation}
This equation has a unique root in $0<|j|<1$, as shown in Appendix~\ref{app:critical-spin}, and gives $|j_{\rm c}|\approx0.2087018708$. Thus, the quoted number is not a universal relation between $a$ and $J_{\rm ADM}$. It is the critical point along the explicitly chosen path $\hat J=\hat a$.

\begin{figure}[H]
  
  \begin{subfigure}[b]{0.45\textwidth}
    \centering
    \includegraphics[width=\textwidth]{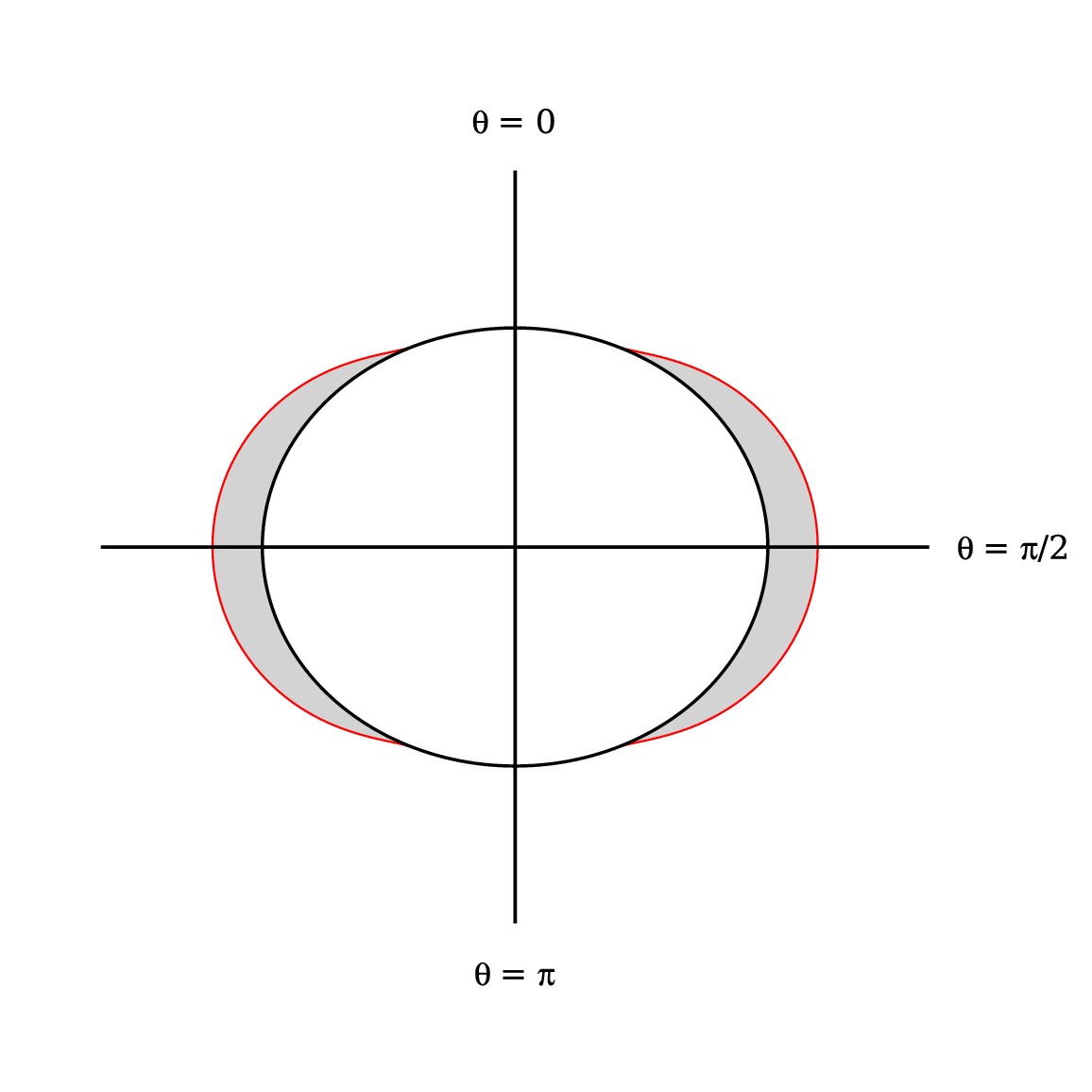}
    \caption{$\hat a=\hat J=0.5$}
  \end{subfigure}
  \hfill
  \begin{subfigure}[b]{0.45\textwidth}
    \centering
    \includegraphics[width=\textwidth]{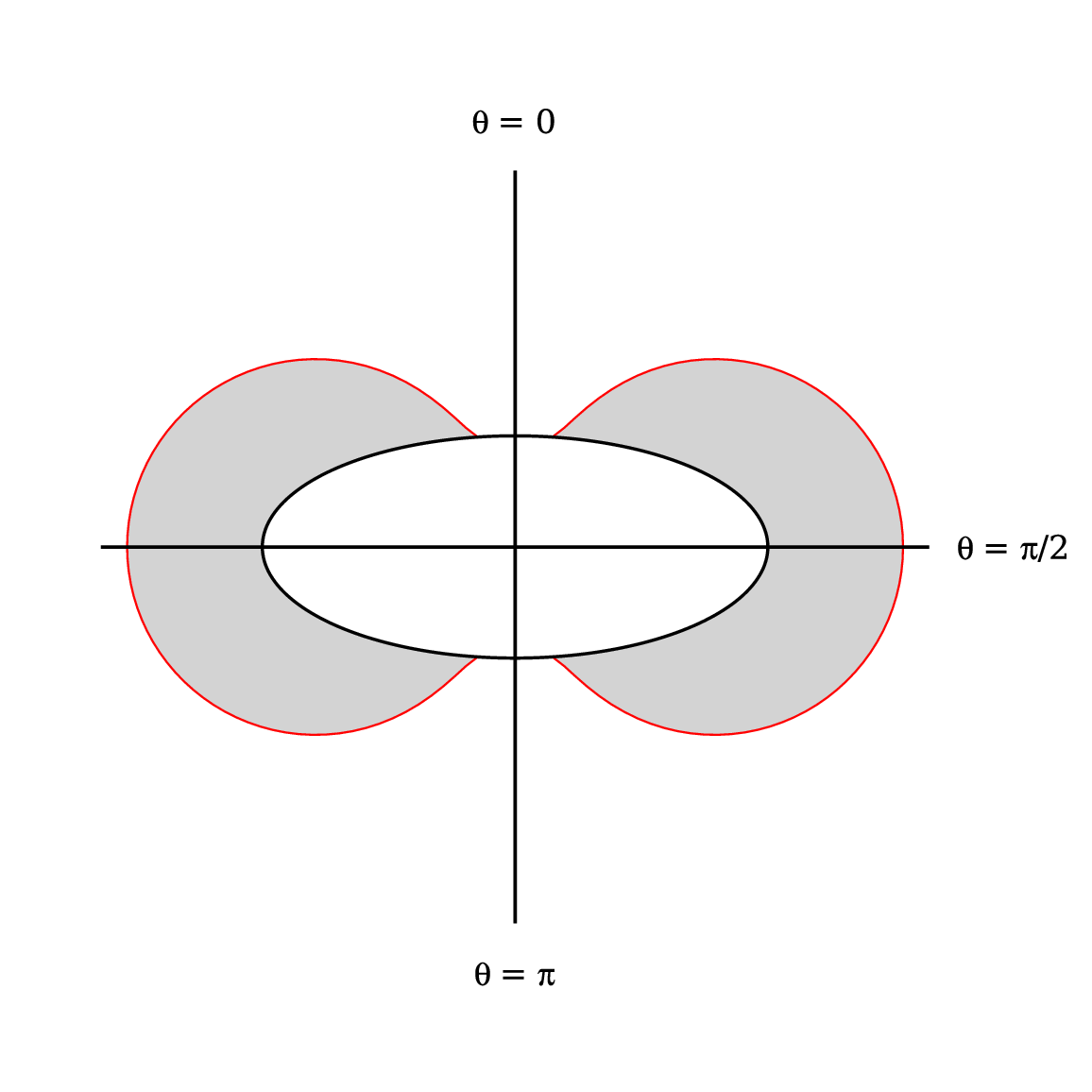}
    \caption{$\hat a=\hat J=0.9$}
  \end{subfigure}
  \caption{Meridional cross sections of the wormhole throat in spheroidal coordinates. The solid black curve denotes the throat boundary and the solid red curve the ergosurface. Both panels lie on the illustrative one-parameter slice $\hat J=\hat a$.}
  \label{fig:ergoregion}
\end{figure}

\subsection{Equatorial Capture Interval}

Let $\mathcal S\subset\mathbb R^2_{(\alpha,\beta)}$ denote the complete two-dimensional set of dark directions on the sky of a distant observer under the illumination prescription stated below. Because the Hamilton--Jacobi equation is not separable away from the reflection plane, this subsection determines only $\mathcal S_{\rm eq}:=\mathcal S\cap\{\beta=0\}$, for an observer in the equatorial plane. Thus, the result is a one-dimensional capture interval on the horizontal celestial axis. It is not a closed shadow contour, and no claim is made here about the vertical extent, area, oblateness, or off-equatorial boundary of $\mathcal S$. The analytic equations are kept in the general form with independent $a$ and $J_{\rm ADM}$. The restriction $J_{\rm ADM}=ar_0$ is imposed only when evaluating Table~\ref{tab:shadow-observables} and the two comparison figures. The use of exterior unstable photon orbits and marginal throat trajectories as capture separatrices follows the established Teo-wormhole optical framework \cite{Shaikh2018PRD}. What is new here is their closed analytic specialisation and branch classification for the present nonseparable oblate geometry. The motion of test particles is governed by the Hamilton--Jacobi equation
\begin{equation}\label{hamiltonJacobi}
\frac{\partial S}{\partial \lambda}=-\frac{1}{2}g^{\mu\nu}\frac{\partial S}{\partial x^\mu}\frac{\partial S}{\partial x^\nu},
\end{equation}
where $\lambda$ is an affine parameter along the geodesic, $g_{\mu\nu}$ denotes the metric tensor, and $S$ is the Jacobi action. Since the spacetime is stationary and axisymmetric, the geodesic motion admits two conserved quantities: the particle energy $E$ and the angular momentum $L$ associated with rotations around the symmetry axis. For the rotating wormhole metric \eqref{rotmet}, the Hamilton--Jacobi equation is not separable in full generality. Although the Hamilton--Jacobi equation is not separable in the full four-dimensional geometry, equatorial null geodesics form a dynamically consistent sector of the geodesic flow.
\begin{lemma}
The hypersurface $\chi=0$ is a totally geodesic invariant submanifold of the
spacetime \eqref{exactrotmet}. Consequently, any geodesic initially tangent to
$\chi=0$ remains confined to $\chi=0$.
\end{lemma}

\begin{proof}
The metric coefficients of \eqref{exactrotmet} are even functions of $\chi$. Indeed, they depend on $\chi$ only through $\chi^2$ and $1-\chi^2$. Therefore, the spacetime is invariant under the discrete reflection $\chi\longmapsto -\chi$. The fixed-point set of this reflection is the hypersurface $\chi=0$. Since the reflection is an isometry, its fixed-point set is totally geodesic. Equivalently, one may verify this directly from the geodesic equation for $\chi$. If a geodesic satisfies $\chi(0)=0$, and $\dot{\chi}(0)=0$, then all terms in the $\chi$-equation of motion vanish at $\chi=0$, because $
\left.\partial_\chi g_{\mu\nu}\right|_{\chi=0}=0$. Hence $\ddot{\chi}(0)=0$. By the uniqueness of solutions to the geodesic equations, the solution with initial data tangent to $\chi=0$ remains on $\chi=0$ for all values of the affine parameter. Thus $\chi=0$ is an invariant totally geodesic submanifold. $\square$
\end{proof}
Therefore, in order to determine $\mathcal S_{\rm eq}$ rather than the full shadow $\mathcal S$, we restrict the analysis to the equatorial plane $\chi=0$ and set the derivatives with respect to $\chi$ to zero. Moreover, we take the reduced Jacobi action of the form
\begin{equation}\label{ansatzS}
	S = \frac{1}{2}\mu^2\lambda - Et + L\varphi + S_r(r),
\end{equation}
where $\mu$ is the rest mass of the test particle. Substituting the ansatz \eqref{ansatzS} into \eqref{hamiltonJacobi}, one obtains the following equatorial geodesic equations
\begin{eqnarray}
\frac{\mathrm d t}{\mathrm d\lambda}&=&\frac{E-\omega L}{N^2},\label{tdotGeo}\\
\frac{\mathrm d\varphi}{\mathrm d\lambda}&=&\frac{L}{\mathcal{K}}+\frac{\omega(E-\omega L)}{N^2},\label{phidotGeo}\\
\left(\frac{\mathrm d r}{\mathrm d\lambda}\right)^2&=&\frac{\Delta}{\Sigma}
\left[\frac{(E-\omega L)^2}{N^2}-\frac{L^2}{ \mathcal{K}}-\mu^2\right].\label{rdotGeo}
\end{eqnarray}
Since we are interested in photon motion, we set $\mu=0$. We also introduce the impact parameter $\xi=L/E$, together with the rescaled affine parameter $\widetilde{\lambda}=E\lambda$. In this way, the explicit dependence on the photon energy is removed, and the null geodesics are parametrised only by the impact parameter $\xi$. The radial equation can then be written in the energy-like form
\begin{equation}\label{radialEffectiveEq}
\left(\frac{\mathrm d r}{\mathrm d\widetilde{\lambda}}\right)^2 + V_{\rm eff}(r;\xi) = 0,
\end{equation}
where the effective potential is
\begin{equation}\label{Veff}
V_{\rm eff}(r;\xi) =\frac{\Delta}{\Sigma} R(r;\xi), \quad R(r; \xi) = \frac{\xi^2}{\mathcal{K}} - \frac{(1-\omega\xi)^2}{N^2}.
\end{equation}
As shown in Sec.~\ref{subsec:proper-radial}, the regular two-ended extension contains two asymptotic regions that are causally connected through the throat. One of these regions is assumed to be illuminated by an external light source, whereas no luminous sources are placed near the throat or in the opposite asymptotic region. Photons emitted in the illuminated region can then follow two qualitatively different classes of trajectories. Some photons are scattered by the wormhole geometry and return to the observer's asymptotic region, while others reach the throat and cross to the other side. A distant observer located in the illuminated region receives only the scattered photons. The photons that cross the throat do not return to that observer and therefore generate a dark interval on the observer's sky. This dark interval, seen against the luminous background, is identified as the wormhole's equatorial shadow slice $\mathcal S_{\rm eq}$. Its two endpoints, not a two-dimensional shadow boundary, are determined by the critical photon trajectories separating scattering orbits from throat-crossing orbits. A photon is scattered back to infinity only if its radial motion admits a turning point, namely
\begin{equation}
\frac{\mathrm d r}{\mathrm d\widetilde{\lambda}}=0 .
\end{equation}
One endpoint of the critical capture interval is generated by an unstable circular photon orbit in the equatorial plane. It is obtained by imposing the circular-orbit conditions away from the throat,
\begin{equation}\label{photonOrbitConditions1}
V_{\rm eff}(r;\xi)=0,\qquad
\frac{\mathrm d V_{\rm eff}}{\mathrm d r}(r;\xi)=0,\qquad
\frac{\mathrm d^2 V_{\rm eff}}{\mathrm d r^2}(r;\xi)\leq 0 .
\end{equation}
Equivalently, since \(\Delta>0\) outside the throat, these conditions can be written in terms of \(R(r;\xi)\) as
\begin{equation}\label{photonOrbitConditionsR}
R(r;\xi)=0,\qquad
\frac{\mathrm dR}{\mathrm dr}(r;\xi)=0,\qquad \frac{\mathrm d^2 R}{\mathrm d r^2}(r;\xi)\leq 0.
\end{equation}
The other endpoint is not generated by a second circular photon orbit. Instead, it is associated with the marginal throat-crossing trajectory and satisfies
\begin{equation}\label{throatBoundaryCondition}
  R(r_{\rm th};\xi)=0, \qquad \frac{\mathrm d R}{\mathrm d \ell}(0;\xi)\leq 0,
\end{equation}
for photons arriving from the $\ell > 0$ exterior. Thus, the two endpoints of the equatorial shadow slice have different geometrical origins. One is controlled by an unstable circular photon orbit outside the throat, whereas the other is controlled by the boundary between reflected and throat-crossing rays.

\subsubsection{Critical Impact Parameters}

For unstable circular photon orbits that do not coincide with the throat, namely for orbits located away from the surface $\Delta=0$, the factor $\Delta/\Sigma$ in the effective potential does not vanish. For the explicit rotating solution \eqref{exactrotmet}, the first two conditions, $R=0$ and $\mathrm dR/\mathrm dr=0$, yield two branches for the critical impact parameter
\begin{equation}\label{xipm}
  \xi_\pm=\left[\left(\omega(r) \pm \frac{1}{\sqrt{\mathcal{K}(r)}}\right)^{-1}\right]_{r=r_{\rm ph}}.
\end{equation}
Here $r_{\rm ph}$ denotes the radius of the circular photon orbit. At this point, a comment on the branch structure is necessary. The condition $R(r;\xi)=0$ gives
\begin{equation}\label{branchCondition}
1-\omega(r)\xi=s\,\frac{\xi}{\sqrt{\mathcal{K}(r)}},\qquad s=\pm1.
\end{equation}
However, the two signs in Eq.~\eqref{branchCondition} do not, in general, correspond to two circular photon orbits at the same radius. Indeed, imposing the second circularity condition, $dR/dr=0$, gives
\begin{equation}\label{branchDerivativeCondition}
-\frac{1}{\mathcal{K}^2(r)}\frac{d \mathcal{K}}{dr}+\frac{2s}{\sqrt{\mathcal{K}(r)}}\frac{d\omega}{dr}=0.
\end{equation}
Using
\begin{equation}\label{omegaprimeBranch}
\frac{d\omega}{dr}=-\frac{6J_{\rm ADM}}{\mathcal{K} ^{3/2}(r)\sqrt{\Delta(r)}},
\end{equation}
one finds that, for $J_{\rm ADM}>0$, only the branch with $s=-1$ satisfies the circular orbit condition. Conversely, for $J_{\rm ADM}<0$, the admissible circular branch is the one with $s=+1$. Squaring Eq.~\eqref{branchDerivativeCondition} yields the algebraic equation for $r_{\rm ph}$, but this operation removes the sign information and must not be interpreted as producing two independent circular orbit branches for a fixed orientation of the ADM angular momentum. For $J_{\rm ADM}>0$, we denote by $\xi_-$ the critical impact parameter associated with the circular photon orbit
\begin{equation}\label{xiRightEdge}
\xi_-=\left[\left(\omega-\frac{1}{\sqrt{\mathcal{K}}}\right)^{-1}\right]_{r=r_{\rm ph}}.
\end{equation}
The opposite endpoint of the equatorial capture interval is not generated by another circular photon orbit at the same radius. Instead, it is associated with photon trajectories that reach the wormhole throat and cross to the second asymptotic region. The condition for such a marginal throat-crossing trajectory is most clearly formulated using the proper radial coordinate introduced in Sec.~\ref{subsec:proper-radial}. On the equatorial plane, one has
\begin{equation}
r^2=\ell^2+r_{\rm th}^2,\qquad
\Delta=\ell^2,\qquad
\frac{dr}{d\widetilde\lambda}=\frac{\ell}{r}\frac{d\ell}{d\widetilde\lambda}.
\end{equation}
Since $\Sigma=r^2$ on the equatorial plane, the radial null equation \eqref{radialEffectiveEq} is equivalent, for $\ell\ne 0$, to
\begin{equation}\label{ellRadialEquation}
\left(\frac{d\ell}{d\widetilde\lambda}\right)^2=(1-\omega\xi)^2-\frac{\xi^2}{\mathcal K}=
-R(r(\ell);\xi).
\end{equation}
By continuity, \eqref{radialEffectiveEq} also gives the regular throat crossing condition at $\ell=0$. A photon crosses the throat if $-R(r_{\rm th};\xi)>0$, whereas the marginal ray separating crossing orbits from reflected orbits satisfies \eqref{throatBoundaryCondition}. Solving this equation gives the second endpoint of $\mathcal S_{\rm eq}$ for $J_{\rm ADM}>0$. The two algebraic values obtained from the marginal throat condition are
\begin{equation}\label{xiLeftEdge}
\xi^{\rm th}_\pm
=
\frac{r_0a^2}
{6J_{\rm ADM}\left[
K\left(\frac{a}{r_0}\right)-E\left(\frac{a}{r_0}\right)
\right]\pm a^2},
\end{equation}
For $a=0$, Eq.~\eqref{xiLeftEdge} is understood by continuity. Moreover, using $K(k)-E(k)=\pi k^2/4+\mathcal O(k^4)$ gives $\xi^{\rm th}_\pm=r_0/[\tfrac{3\pi}{2}\hat J\pm1]$, where $K(k)$ and $E(k)$ denote the complete elliptic integrals of the first and second kind, respectively. We now determine which of these two branches is compatible with the marginal throat-crossing condition. It is preferable to perform this sign analysis in the regular throat coordinate $\ell$, rather than in the one-sided coordinate $r$. Indeed, $d\omega/dr$ diverges as $r\to r_{\rm th}^{+}$, whereas $\Omega'(\ell)$ is finite at $\ell=0$. On the equatorial plane, the reduced radial function can be written as
\begin{equation}
R(\ell;\xi)
=
\frac{\xi^2}{\ell^2+r_0^2}
-
\left[1-\Omega(\ell)\xi\right]^2 .
\end{equation}
Differentiating with respect to \(\ell\) gives
\begin{equation}
\frac{dR}{d\ell}
=
-\frac{2\ell\xi^2}{(\ell^2+r_0^2)^2}
+
2\xi\Omega'(\ell)\left[1-\Omega(\ell)\xi\right].
\end{equation}
At the throat, \(\ell=0\), this reduces to
\begin{equation}\label{dRdellThroat}
\left.
\frac{dR}{d\ell}
\right|_{\ell=0}
=
2\xi\Omega'(0)\left(1-\Omega_{\rm th}\xi\right),
\end{equation}
where $\Omega_{\rm th}:=\Omega(0)$. The marginal condition \(R(0;\xi)=0\) implies
\begin{equation}\label{branchRelationThroat}
1-\Omega_{\rm th}\xi
=
s\,\frac{\xi}{r_0},
\qquad
s=\pm1 .
\end{equation}
Substituting \eqref{branchRelationThroat} into \eqref{dRdellThroat}, we find
\begin{equation}
\left.
\frac{dR}{d\ell}
\right|_{\ell=0}
=
\frac{2s\xi^2}{r_0}\Omega'(0).
\end{equation}
Throughout this discussion, we assume $J_{\rm ADM}>0$. This fixes the orientation of the frame dragging independently of the sign convention for the oblate parameter $a$. Since
\begin{equation}
\Omega'(0)=-\frac{6J_{\rm ADM}}{r_0^3\sqrt{r_0^2-a^2}}<0,
\end{equation}
the condition
\begin{equation}
\left.\frac{dR}{d\ell}
\right|_{\ell=0}\leq 0
\end{equation}
selects the branch $s=+1$. Therefore, for positive ADM angular momentum, the admissible marginal throat-crossing value is
\begin{equation}
\xi^{\rm th}_+=\frac{1}{\Omega_{\rm th}+1/r_0}.
\end{equation}
Using
\begin{equation}
\Omega_{\rm th}=\frac{6J_{\rm ADM}}{r_0a^2}\left[K\left(\frac{a}{r_0}\right)-E\left(\frac{a}{r_0}\right)\right],
\end{equation}
this gives
\begin{equation}
\xi^{\rm th}_+=\frac{r_0a^2}{6J_{\rm ADM}\left[
K\left(\frac{a}{r_0}\right)-E\left(\frac{a}{r_0}\right)
\right]+a^2}.
\end{equation}
The other algebraic branch, corresponding to $s=-1$, gives
\begin{equation}
\left.\frac{dR}{d\ell}\right|_{\ell=0}>0,
\end{equation}
and is therefore incompatible with marginal throat crossing from the exterior side. In what follows, we denote the relevant throat-induced equatorial endpoint simply by $\xi_+\equiv \xi^{\rm th}_+$. $r_{\rm ph}$ is determined by the algebraic equation
\begin{equation}\label{rphEquation}
36J_{\rm ADM}^2=r_{\rm ph}^2\left(r_{\rm ph}^2+a^2-r_0^2\right).
\end{equation}
The positive root is
\begin{equation}\label{rphCorrected}
r_{\rm ph}=\left[\frac{1}{2}\left(r_0^2-a^2+\sqrt{(r_0^2-a^2)^2+144J_{\rm ADM}^2}\right)\right]^{1/2}.
\end{equation}
Since $r_{\rm th}^2=r_0^2-a^2$, it follows that
\begin{equation}
r_{\rm ph}^2-r_{\rm th}^2=\frac{1}{2}\left[\sqrt{(r_0^2-a^2)^2+144J_{\rm ADM}^2}-(r_0^2-a^2)\right]\geq 0.
\end{equation}
The inequality is strict whenever $J_{\rm ADM}\neq0$. Thus, the exterior circular orbit is produced by nonzero angular momentum and lies outside the throat. This conclusion does not require $J_{\rm ADM}=ar_0$. For the branch relevant when $J_{\rm ADM}>0$, the second derivative of the reduced radial function is
\begin{equation}\label{RsecondDerivativePhoton}
\left.\frac{d^2R}{dr^2}\right|_{r=r_{\rm ph}}=-\frac{\xi_-^2\left(r_{\rm ph}^4+36J_{\rm ADM}^2\right)}{18J_{\rm ADM}^2\left(r_{\rm ph}^2+a^2\right)^2}<0,\qquad 
J_{\rm ADM}\neq0.
\end{equation}
Therefore, the circular photon orbit is unstable.

\subsubsection{Connectedness of the Capture Interval}

We establish that the set of equatorial impact parameters leading to capture by the wormhole is connected. This point is important because it ensures that the equatorial shadow slice is bounded by only two critical values of \(\xi\), rather than by several disconnected capture windows. On the equatorial plane, and for the explicit rotating solution considered here, the null radial equation may be written in the form
\begin{equation}\label{globalRadialEquation}
\left(\frac{dr}{d\widetilde\lambda}\right)^2=-\frac{\Delta}{\Sigma}\,R(r;\xi),\qquad
R(r;\xi)=\frac{\xi^2}{r^2+a^2}-\left[1-\omega(r)\xi\right]^2 .
\end{equation}
In the exterior region $r>r_{\rm th}$, the prefactor $\Delta/\Sigma$ is strictly positive. Hence, the physically accessible part of the radial motion is characterised by $R(r;\xi)\leq 0$, and radial turning points outside the throat are precisely the zeros of $R(r;\xi)$. Moreover, since the spacetime is asymptotically flat and $\omega(r)\to 0$ as $r\to\infty$, the reduced radial function satisfies $R(r;\xi)=-1+\mathcal O(r^{-2})$ as $r\to\infty$. Thus, photons arriving from the observer's asymptotic region initially lie in the allowed region of the radial motion. Such a photon is reflected back to the same asymptotic region if and only if, along its inward trajectory, the function $R(r;\xi)$ reaches zero at some radius $r>r_{\rm th}$. Conversely, if no such exterior zero exists, the radial motion remains allowed all the way down to the throat, and the photon crosses into the second asymptotic region. The boundary between reflected and throat-crossing trajectories can therefore arise only when the first exterior zero of $R(r;\xi)$ is created or destroyed. For a fixed regular pair $|\hat a|<1$ and $\hat J\neq0$, this can occur in precisely two geometrically distinct ways. The first possibility is that the first zero occurs at an interior degenerate root, that is
\begin{equation}\label{globalInteriorBoundary}
R(r;\xi)=0,\qquad
\partial_r R(r;\xi)=0,\qquad
r>r_{\rm th}.
\end{equation}
This is the usual circular-photon-orbit condition. The additional inequality $\partial_r^2R(r;\xi)<0$, verified above for the relevant branch, identifies this orbit as unstable and, therefore, as a genuine separatrix between captured and scattered null rays. The second possibility is that the first zero is not created in the interior of the exterior domain, but instead is pushed to the endpoint of the radial interval. In this case, the critical ray is marginally allowed at the throat and satisfies $R(r_{\rm th};\xi)=0$, with the corresponding one-sided condition, most naturally expressed in the regular proper radial coordinate $\ell$, selecting the branch that is reached by photons arriving from the $\ell>0$ exterior. We determined the global structure of the boundary conditions by solving the circular-orbit system \eqref{globalInteriorBoundary} for all admissible exterior radii $r\in(r_{\rm th},\infty)$. The endpoint case was then obtained independently from the throat condition. For each regular pair $(\hat a,\hat J)$ with $\hat J\neq0$, the algebraic classification leaves a single exterior unstable circular branch contributing to the capture boundary. The other endpoint is supplied by the marginal throat-crossing branch. After excluding spurious roots, interior double roots, and nonphysical throat-compatible branches, no additional admissible boundary branches remain. Hence, the equatorial capture set forms one connected interval in impact-parameter space, with endpoints given by the exterior circular photon orbit and the admissible throat-crossing value $\mathcal C_{\rm cap}=[\,\xi_-,\xi_+\,]$ or $\mathcal C_{\rm cap}=[\,\xi_+,\xi_-\,]$, depending on the orientation convention adopted for the ADM angular momentum and for the impact parameter. In either convention, the physical conclusion is the same: the equatorial shadow slice has exactly two endpoints. One endpoint is controlled by the unstable circular photon orbit outside the throat, while the other endpoint is controlled by the marginal trajectory that just reaches the throat.

\subsubsection{Projection onto the equatorial celestial axis}

The relations \eqref{xiRightEdge} and \eqref{xiLeftEdge} therefore determine the two endpoints of the equatorial capture interval in impact-parameter space. However, a distant observer does not measure the impact parameter directly. Instead, the observer sees its projection onto the observer's sky, namely the plane orthogonal to the line of sight and passing through the observer. We denote the corresponding celestial coordinate by $\alpha$. This coordinate specifies the apparent angular position of the image on the observer's sky. In terms of the photon trajectory, it is defined by \cite{bray1986kerr}
\begin{equation}\label{alphaDef}
\alpha = \lim_{r\to\infty} \left(r^2\sin\theta_0 \frac{\mathrm d\varphi}{\mathrm dr} \right),
\end{equation}
where $\theta_0$ is the inclination angle between the wormhole rotation axis and the observer's line of sight. Since our analysis is restricted to the equatorial plane, we take $\theta_0=\pi/2$. Using the equatorial geodesic equations \eqref{tdotGeo}--\eqref{rdotGeo} in the asymptotic region, one then obtains the simple relation
\begin{equation}\label{alphaXiRelation}
\alpha_\pm=-\xi_\mp.
\end{equation}
The equatorial shadow slice is therefore the line segment $\mathcal S_{\rm eq}=\{(\alpha,\beta):\ \beta=0,\ \alpha\in[\alpha_-,\alpha_+]\}$, with the order of the endpoints reversed if the opposite orientation convention is adopted. For later tabulation, we define only one-dimensional slice diagnostics, namely
\begin{equation}\label{equatorialSliceDiagnostics}
\alpha_{\rm c}=\frac{\alpha_++\alpha_-}{2},\qquad
R_{\rm eq}=\frac{|\alpha_+-\alpha_-|}{2}.
\end{equation}
Neither $\alpha_{\rm c}$ nor $R_{\rm eq}$ should be interpreted as the centre or radius of a full two-dimensional shadow. Here $\alpha_+$ and $\alpha_-$ denote the right and left endpoints, respectively, of the horizontal equatorial slice on the observer's sky. The sign convention is chosen to match the standard observer-sky convention,	in which photon trajectories are traced backwards from a distant observer toward the compact object. The one-dimensional set $\mathcal S_{\rm eq}$ observed at infinity is therefore obtained by mapping the critical impact parameters $\xi_\pm$ to the celestial coordinates $\alpha_\pm$ according to Eq.~\eqref{alphaXiRelation}. This mapping does not determine the boundary for $\beta\neq 0$.

\begin{table}[H]
\centering
\caption{One-dimensional equatorial-slice observables on the illustrative parameter line $\hat a=\hat J\equiv j$, with $x=r/r_0$. The exact family permits independent $\hat a$ and $\hat J$. The restriction is made here only to display a one-dimensional sequence. The quantities $x_{\rm th}$ and $x_{\rm ph}$ are the dimensionless coordinate radii of the throat representative and photon orbit, while $\alpha_\pm/r_0$, $\alpha_{\rm c}/r_0$, and $R_{\rm eq}/r_0$ are the two celestial endpoints, midpoint, and half-width defined in Eq.~\eqref{equatorialSliceDiagnostics}. They are not full-shadow observables.}
\label{tab:shadow-observables}
\begin{tabular}{cccccccc}
\toprule
$j$
& $x_{\rm th}$
& $x_{\rm ph}$
& $\alpha_-/r_0$
& $\alpha_+/r_0$
& $\alpha_{\rm c}/r_0$
& $R_{\rm eq}/r_0$\\
\midrule
$0.005$ & 0.999987  & 1.00043  & -0.976980 & 1.02366 &  0.023340 & 1.00032  \\
$0.3$   & 0.953939 & 1.52040 & -0.405794 & 2.19552 & 0.894863 & 1.30066 \\
$0.6$   & 0.800000 & 1.98348 & -0.231123 & 3.00292 & 1.38590 & 1.61702 \\
$0.9$   & 0.435890 & 2.34432 & -0.119156 & 3.65464 & 1.76774 & 1.88690 \\
\bottomrule
\end{tabular}
\end{table}

\begin{figure}[H]
  \begin{subfigure}[b]{0.45\textwidth}
    \centering
    \includegraphics[width=\textwidth]{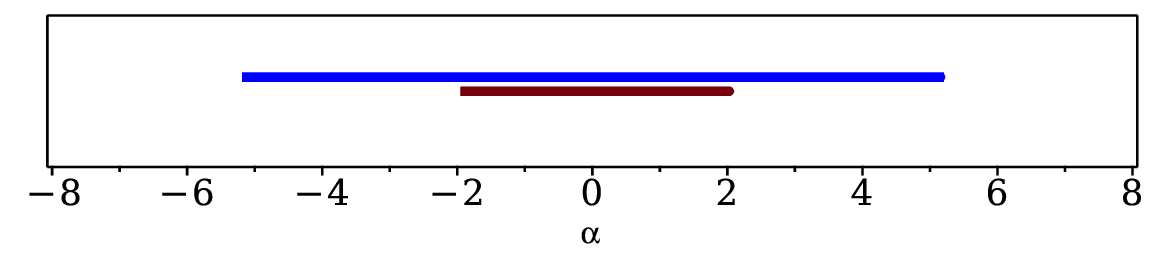}
    \caption{$\hat a=\hat J=j=0.005$}
  \end{subfigure}
  \hfill
  \begin{subfigure}[b]{0.45\textwidth}
    \centering
    \includegraphics[width=\textwidth]{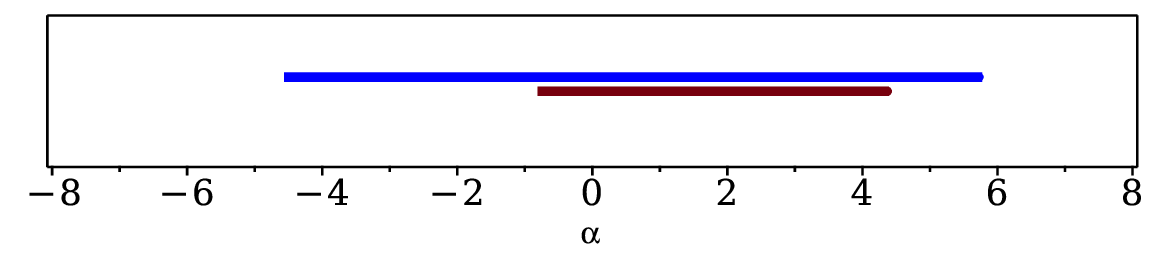}
    \caption{$\hat a=\hat J=j=0.3$}
  \end{subfigure}
  \hfill

  \vspace{1em}

  \begin{subfigure}[b]{0.45\textwidth}
    \centering
    \includegraphics[width=\textwidth]{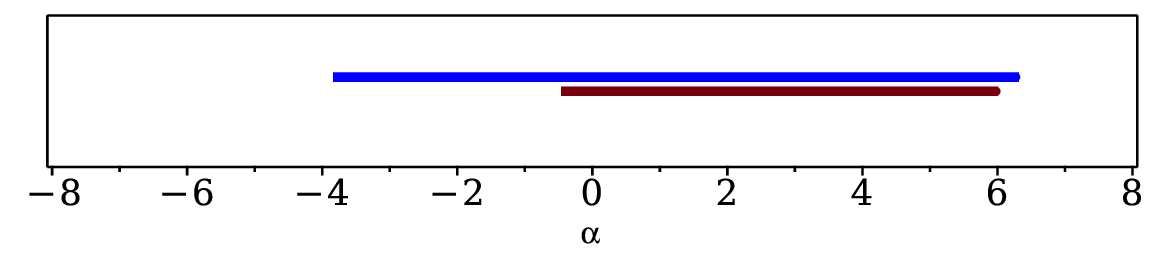}
    \caption{$\hat a=\hat J=j=0.6$}
  \end{subfigure}
  \hfill
  \begin{subfigure}[b]{0.45\textwidth}
    \centering
    \includegraphics[width=\textwidth]{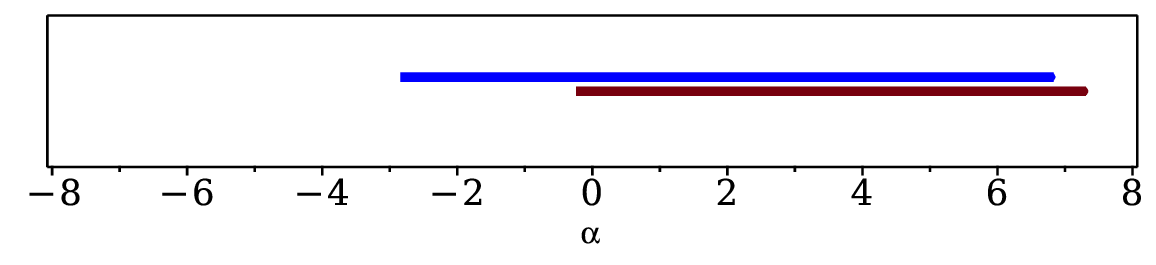}
    \caption{$\hat a=\hat J=j=0.9$}
  \end{subfigure}

  \caption{One-dimensional equatorial capture intervals $\mathcal S_{\rm eq}$ of the wormhole (solid red line) and the corresponding horizontal slices of the Kerr shadow (solid blue line). For the wormhole the displayed sequence is the line $\hat a=\hat J\equiv j$, for Kerr, $j=J_{\rm ADM}/M^2$. We set $M=1$ and $r_0=2M$, matching the Schwarzschild horizon radius in the static comparison. The celestial coordinates are in units of $M$. These line segments are not complete two-dimensional shadow contours.}
  \label{fig:shadow-match-horizon}
\end{figure}

\begin{figure}[H]
  \begin{subfigure}[b]{0.45\textwidth}
    \centering
    \includegraphics[width=\textwidth]{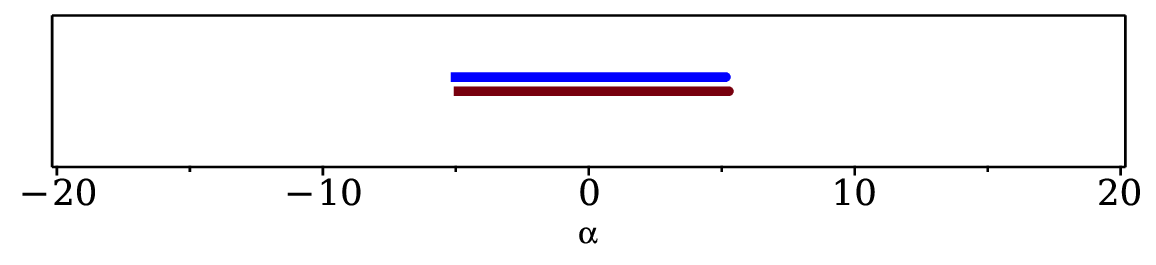}
    \caption{$\hat a=\hat J=j=0.005$}
  \end{subfigure}
  \hfill
  \begin{subfigure}[b]{0.45\textwidth}
    \centering
    \includegraphics[width=\textwidth]{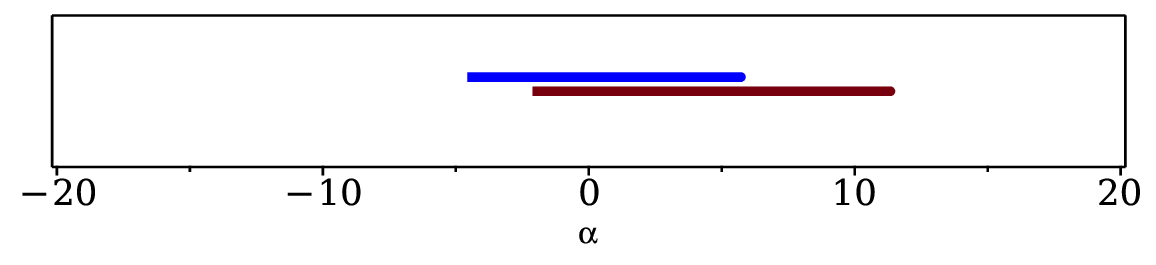}
    \caption{$\hat a=\hat J=j=0.3$}
  \end{subfigure}
  \hfill

  \vspace{1em}

  \begin{subfigure}[b]{0.45\textwidth}
    \centering
    \includegraphics[width=\textwidth]{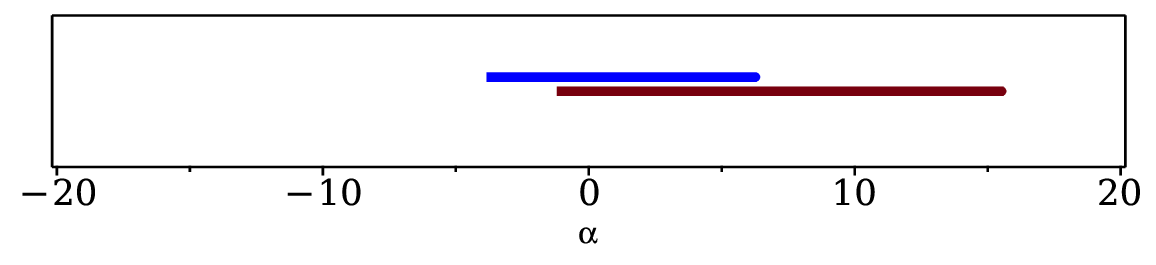}
    \caption{$\hat a=\hat J=j=0.6$}
  \end{subfigure}
  \hfill
  \begin{subfigure}[b]{0.45\textwidth}
    \centering
    \includegraphics[width=\textwidth]{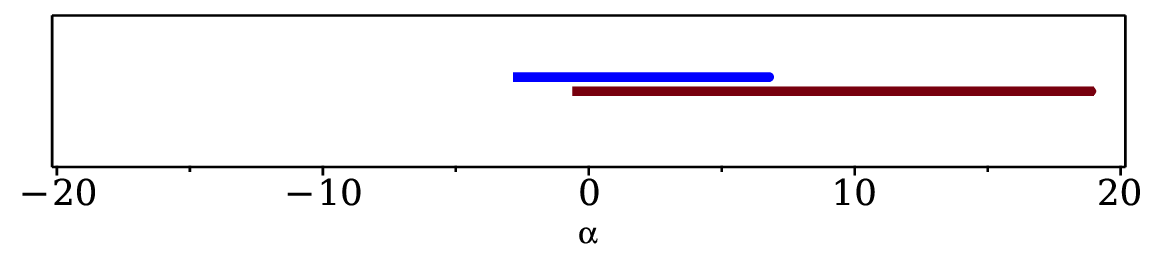}
    \caption{$\hat a=\hat J=j=0.9$}
  \end{subfigure}

  \caption{One-dimensional equatorial capture intervals $\mathcal S_{\rm eq}$ of the wormhole (solid red line) and the corresponding horizontal slices of the Kerr shadow (solid blue line). For the wormhole the displayed sequence is the line $\hat a=\hat J\equiv j$; for Kerr, $j=J_{\rm ADM}/M^2$. We set $M=1$ and $r_0=\sqrt{27}$ so that the two static capture intervals coincide. The celestial coordinates are in units of $M$. These line segments are not complete two-dimensional shadow contours.}
  \label{fig:shadow-match-static}
\end{figure}

\section{Conclusions and Outlook}

We have constructed an exact Kerr-like rotating extension of the zero-redshift Morris--Thorne wormhole within the assumptions $N=1$, $\mathcal K=\mathcal K(r)$, and $\omega=\omega(r)$. The field equations fix $\mathcal K=r^2+a^2$ and the radial form of $\omega$, while asymptotic flatness fixes its amplitude in terms of the independently measured ADM angular momentum $J_{\rm ADM}$. Because $M_{\rm ADM}=0$, neither the Kerr relation nor any equation derived here identifies $a$ with $J_{\rm ADM}$. The exact solution is a family $(r_0,a,J_{\rm ADM})$. In this interpretation, $a$ controls the oblate spatial structure and $J_{\rm ADM}$ controls the magnitude and orientation of frame dragging. The relation $J_{\rm ADM}=ar_0$ is only the one-parameter path used for the displayed numerical sequences. The principal geometrical results are most naturally stated without reference to a radial coordinate. In the canonical stationary foliation, the throat is the unique member $S_0$ of the closed family $S_\ell$ with vanishing mean curvature, positive area second variation, both future null expansions equal to zero, and positive normal flare-out derivative. Its equatorial proper circumference is $2\pi r_0$. The equations $\ell=0$, $\Delta=0$, and $r_{\rm th}=\sqrt{r_0^2-a^2}$ merely represent this surface in the chosen charts. It is scalar-polynomially regular for $|a|<r_0$ and finite $J_{\rm ADM}$. The signed coordinate $\ell$ yields a smooth two-ended extension. Moreover, $g^{tt}=-1$ makes $t$ a global temporal function, and the throat world tube is timelike rather than null, signaling that causal curves can cross it. None of these statements uses $J_{\rm ADM}=ar_0$. The effects that involve frame dragging depend on both dimensionless parameters $\hat a=a/r_0$ and $\hat J=J_{\rm ADM}/r_0^2$. The ergoregion onset is the curve \eqref{criticalErgoCurve}. The previously quoted value $|j_{\rm c}|\simeq0.2087018708$ is its unique intersection with $\hat J=\hat a\equiv j$. The two endpoints of the equatorial capture interval are also two-parameter quantities. In particular, the exterior circular orbit satisfies $r_{\rm ph}^2(r_{\rm ph}^2+a^2-r_0^2)=36J_{\rm ADM}^2$, and the marginal throat endpoint is given by Eq.~\eqref{xiLeftEdge}. Table~\ref{tab:shadow-observables} and Figs.~\ref{fig:shadow-match-horizon} and \ref{fig:shadow-match-static} specialise these general expressions to the illustrative one-parameter line. They display $\mathcal S_{\rm eq}\subset\{\beta=0\}$ only and do not determine a complete two-dimensional shadow contour. Relative to the existing literature, the present paper should therefore be read as a complementary exact construction rather than as the first rotating wormhole. The Teo ansatz, rotating phantom and electromagnetic wormholes, Kerr-like black-bounce metrics, minimal-surface throat criteria, and photon-separatrix methods are established ingredients. Within the restricted ansatz used here, the new results are the field-equation determination of $\mathcal K=r^2+a^2$, the elliptic-integral frame-dragging profile for $b=r_0^2/r$, the independent oblateness and ADM-angular-momentum parameters, the quasi-local characterisation of the oblate throat, the two-parameter ergoregion threshold, and the analytic equatorial capture endpoints. Compared specifically with our earlier exact spinning Morris--Thorne solution \cite{Batic2026EPJC}, the present geometry replaces the spherical throat and elementary dragging profile by an oblate minimal throat and elliptic dragging profile. Conversely, the loss of separability means that the present paper does not reproduce the full-shadow or multipole analysis of that earlier work. The supporting tensor is a conserved anisotropic stress tensor reconstructed from the geometry. The equatorial radial NEC and WEC violations persist in the unrestricted family and therefore do not arise from the phenomenological identification of the oblateness and angular-momentum parameters. At the same time, the reconstruction does not supply a microscopic matter action or a realistic equation of state. We consequently interpret the source as effective and phenomenological. Establishing a physical completion would require an explicit matter sector whose fields reproduce the tensor while satisfying their own equations of motion, together with a perturbative stability analysis.\\
Natural extensions include exploring the full $(\hat a,\hat J)$ parameter plane numerically, identifying whether a specific exotic-field or effective-gravity model selects a preferred curve $J_{\rm ADM}(a)$, studying global energy-condition domains and integrated exotic-matter measures, and analysing tidal constraints, geodesic completeness, and linear stability. Off-equatorial null dynamics will require numerical ray tracing because the Hamilton--Jacobi equation is not generally separable. Only such a calculation can produce the full two-dimensional shadow and test whether its off-equatorial boundary remains connected. These steps are necessary before the solution can be assigned an astrophysically realistic matter interpretation.

\appendix
\section{Komar angular momentum}\label{app:komar}
We verify that the parameter $J_{\rm ADM}$, which appears in the asymptotic normalisation of the frame-dragging function, coincides with the Komar angular momentum of the spacetime. This provides an independent consistency check of its interpretation as the total angular momentum measured at spatial infinity. For a stationary and axisymmetric spacetime, let $\eta^\alpha=\left(\partial_\varphi\right)^\alpha$ denote the axial Killing vector field. The Komar angular momentum associated with $\eta^\alpha$ is \cite{Wald1986}
\begin{equation}\label{KomarDefinitionNoHodge}
J_{\rm K}=-\frac{1}{16\pi}\lim_{r\to\infty}\int_{\mathcal{S}_r}\nabla^\alpha \eta^\beta\,dS_{\alpha\beta},
\end{equation}
where $\mathcal{S}_r$ is a two-sphere of constant $t$ and constant $r$ in the asymptotic region, and $dS_{\alpha\beta}$ is its oriented surface element. For the rotating Morris--Thorne wormhole, the metric component encoding the asymptotic angular momentum is
\begin{equation}\label{gtphiKomar}
g_{t\varphi}=-(1-\chi^2)\mathcal{K}(r)\omega(r),\qquad
\mathcal{K}(r)=r^2+a^2.
\end{equation}
The frame-dragging function is normalised at large radius according to
\begin{equation}\label{omegaAsymptoticKomar}
\omega(r)=\frac{2J_{\rm ADM}}{r^3}+\mathcal{O}\left(r^{-5}\right),\qquad r\to\infty .
\end{equation}
Therefore,
\begin{equation}\label{gtphiAsymptoticKomar}
g_{t\varphi}=-\frac{2J_{\rm ADM}}{r}(1-\chi^2)+\mathcal{O}\left(r^{-3}\right).
\end{equation}
Returning to the standard polar angle $\chi=\cos\theta$, this becomes
\begin{equation}\label{gtphiAsymptoticTheta}
g_{t\varphi}=-\frac{2J_{\rm ADM}}{r}\sin^2\theta+\mathcal{O}\left(r^{-3}\right).
\end{equation}
On $\mathcal{S}_r$, we choose the oriented surface element
\begin{equation}\label{SurfaceElementKomar}
dS_{\alpha\beta}=2n_{[\alpha}\sigma_{\beta]}\,dA,
\end{equation}
where $n^\alpha$ is the future-directed unit normal to the $t=$const.-hypersurface, $\sigma^\alpha$ is the outward-pointing unit normal to the $r=$const.-surface within that hypersurface, and $dA=r^2\sin\theta\,d\theta d\varphi+\mathcal{O}(1)$ is the induced area element at spatial infinity. As $r\to\infty$, $n_\alpha dx^\alpha=-dt+\mathcal{O}(r^{-1})$, and $\sigma_\alpha dx^\alpha=dr+\mathcal{O}(r^{-1})$. Hence
\begin{equation}\label{SurfaceComponentsKomar}
dS_{tr}=-dA,\qquad dS_{rt}=dA.
\end{equation}
It follows that the Komar integrand receives contributions from both antisymmetric components
\begin{equation}\label{KomarContraction}
\nabla^\alpha\eta^\beta dS_{\alpha\beta}=\left(\nabla^t\eta^r dS_{tr}+\nabla^r\eta^t dS_{rt}\right).
\end{equation}
We now compute the leading asymptotic value of this contraction. Since \(\eta^\alpha = \delta^\alpha_\varphi\), one has $\nabla^\alpha \eta^\beta=g^{\alpha\mu}\Gamma^\beta{}_{\mu\varphi}$. The metric components needed at spatial infinity are \eqref{gtphiAsymptoticTheta} and $g_{\varphi\varphi}=r^2\sin^2\theta+\mathcal{O}(1)$, and therefore
\begin{equation}
g^{t\varphi}=-\frac{2J_{\rm ADM}}{r^3}+\mathcal{O}(r^{-5}),\qquad
g^{tt}=-1,\qquad
g^{rr}=1+\mathcal{O}(r^{-1}).
\end{equation}
Using stationarity and axisymmetry, the relevant Christoffel symbol is
\begin{equation}
\Gamma^t{}_{r\varphi}=\frac{1}{2} g^{tt}\partial_r g_{t\varphi}+\frac{1}{2} g^{t\varphi}\partial_r g_{\varphi\varphi}+\mathcal{O}(r^{-3}).
\end{equation}
Since
\begin{equation}
\partial_r g_{t\varphi}=\frac{2J_{\rm ADM}}{r^2}\sin^2\theta+\mathcal{O}(r^{-4}),\qquad
\partial_r g_{\varphi\varphi}=2r\sin^2\theta+\mathcal{O}(r^{-1}),
\end{equation}
we obtain
\begin{equation}\label{GradretaKomar}
\nabla^r\eta^t=\Gamma^t{}_{r\varphi}=-\frac{3J_{\rm ADM}}{r^2}\sin^2\theta+\mathcal{O}(r^{-3}).
\end{equation}
Similarly, $\nabla^t\eta^r=g^{tt}\Gamma^r{}_{t\varphi}+g^{t\varphi}\Gamma^r{}_{\varphi\varphi}$. At leading order,
\begin{equation}
\Gamma^r{}_{t\varphi}=-\frac{1}{2} g^{rr}\partial_r g_{t\varphi}=
-\frac{J_{\rm ADM}}{r^2}\sin^2\theta+\mathcal{O}(r^{-3}),
\end{equation}
while
\begin{equation}
\Gamma^r{}_{\varphi\varphi}=-\frac{1}{2} g^{rr}\partial_r g_{\varphi\varphi}=-r\sin^2\theta+\mathcal{O}(r^{-1}).
\end{equation}
Therefore,
\begin{equation}\label{GradtetaKomar}
\nabla^t\eta^r=-\left(-\frac{J_{\rm ADM}}{r^2}\sin^2\theta\right)+\left(-\frac{2J_{\rm ADM}}{r^3}\right)\left(-r\sin^2\theta\right)+\mathcal O(r^{-3})
=\frac{3J_{\rm ADM}}{r^2}\sin^2\theta+\mathcal O(r^{-3}).
\end{equation}
Equations \eqref{SurfaceComponentsKomar}--\eqref{GradtetaKomar} then imply
\begin{equation}\label{KomarIntegrandFinal}
\nabla^\alpha\eta^\beta dS_{\alpha\beta}=\left(\nabla^t\eta^r dS_{tr}+\nabla^r\eta^t dS_{rt}\right)
=-\frac{6J_{\rm ADM}}{r^2}\sin^2\theta\,dA+\mathcal{O}(r^{-1})\,d\theta d\varphi .
\end{equation}
Using $dA=r^2\sin\theta\,d\theta d\varphi+\mathcal{O}(1)$ together with \eqref{KomarDefinitionNoHodge}, the Komar angular momentum becomes
\begin{equation}\label{KomarEqualsJADM}
J_{\rm K}=\frac{6J_{\rm ADM}}{16\pi}\int_0^{2\pi}\int_0^\pi
    \sin^3\theta\,d\theta d\varphi=J_{\rm ADM}.
\end{equation}
Thus the parameter $J_{\rm ADM}$ appearing in the large-$r$ fall-off of $g_{t\varphi}$ is precisely the Komar angular momentum of the rotating wormhole spacetime. This confirms that $J_{\rm ADM}$ is the physical angular momentum measured at spatial infinity.

\section{Scalar-polynomial regularity at the throat}
\label{app:regularity}

Curvature invariants provide a coordinate-independent diagnostic of local regularity and are particularly useful in wormhole geometries, where the radial coordinate used in the exterior patch often becomes singular at the throat \cite{Visser1996, Lobo2008CQGR}.  We now show that the rotating metric \eqref{exactrotmet} is scalar-polynomial regular at the throat for $0\leq |a|<r_0$. The endpoint $|a|=r_0$ is not included in this regular family. To this end, it is convenient to use the polar angle $\vartheta$, defined by $\chi=\cos\vartheta$, and to introduce
\begin{equation}
A(r)=r^2+a^2,\qquad
\Sigma=r^2+a^2\cos^2\vartheta,\qquad
\Delta=r^2+a^2-r_0^2.
\end{equation}
For $b(r)=r_0^2/r$, the line element can be written in the orthonormal coframe
\begin{equation}
e^0 = dt,\quad
e^1 = \sqrt{\frac{\Sigma}{\Delta}}\,dr,\quad
e^2 = \sqrt{\Sigma}\,d\vartheta,\quad
e^3 = \sqrt{A}\sin\vartheta\left(d\varphi-\omega dt\right),
\end{equation}
so that $ds^2=-(e^0)^2+(e^1)^2+(e^2)^2+(e^3)^2$.  The only potentially dangerous term is the radial derivative of the frame-dragging function. From \eqref{omegaexact}, or directly from the quadrature representation, one has
\begin{equation}\label{omegaPrimeRegApp}
  \frac{d\omega}{dr}=-\frac{6J_{\rm ADM}}{A^{3/2}\sqrt{\Delta}}.
\end{equation}
Although $d\omega/dr$ diverges as $\Delta^{-1/2}$ at the throat, this divergence is precisely cancelled in the orthonormal frame. Indeed, the combination entering the connection is
\begin{equation}\label{VdefApp}
\mathcal V(r,\vartheta):=\sqrt{A}\sin\vartheta\sqrt{\frac{\Delta}{\Sigma}}\,\frac{d\omega}{dr}=-\frac{6J_{\rm ADM}\sin\vartheta}{A\sqrt{\Sigma}},
\end{equation}
which is finite provided $A>0$ and $\Sigma>0$. Define the following auxiliary functions
\begin{equation}\label{BCEFdefsApp}
B=\frac{r\sqrt{\Delta}}{\Sigma^{3/2}},\quad
C=-\frac{a^2\sin\vartheta\cos\vartheta}{\Sigma^{3/2}},\quad
E=\frac{r\sqrt{\Delta}}{A\sqrt{\Sigma}},\quad
F=\frac{\cot\vartheta}{\sqrt{\Sigma}},
\end{equation}
and denote frame derivatives by
\begin{equation}\label{frameDerivativesApp}
X_1:=e_1(X)=\sqrt{\frac{\Delta}{\Sigma}}\,\partial_r X,
\qquad
X_2:=e_2(X)=\frac{1}{\sqrt{\Sigma}}\,\partial_\vartheta X .
\end{equation}
The nonvanishing spin-connection one-forms $\varpi_{ab}=-\varpi_{ba}$ are
\begin{equation}\label{connectionFormsApp}
\varpi_{01}=-\frac{\mathcal V}{2}\,e^3,\quad
\varpi_{03}=-\frac{\mathcal V}{2}\,e^1,\quad
\varpi_{12}=C\,e^1-B\,e^2,\quad
\varpi_{13}=\frac{\mathcal V}{2}\,e^0-E\,e^3,\quad
\varpi_{23}=-F\,e^3.
\end{equation}
The curvature two-forms
$\mathcal R_{ab}=d\varpi_{ab}+\varpi_a{}^c\wedge\varpi_{cb}$ are then
\begin{subequations}
\begin{align}
\mathcal R_{01}
={}&
-\frac{3}{4}\mathcal V^2\,e^0\wedge e^1
-\left(E\mathcal V+\frac{\mathcal V_1}{2}\right)e^1\wedge e^3
-\frac{1}{2}\left(F\mathcal V+\mathcal V_2\right)e^2\wedge e^3,
\\
\mathcal R_{02}
={}&
\frac{\mathcal V}{2}(C-F)\,e^1\wedge e^3
-\frac{B\mathcal V}{2}\,e^2\wedge e^3,
\\
\mathcal R_{03}
={}&
\frac{\mathcal V^2}{4}\,e^0\wedge e^3
+\frac{1}{2}\left(C\mathcal V+\mathcal V_2\right)e^1\wedge e^2,
\\
\mathcal R_{12}
={}&
-\left(B^2+B_1+C^2+C_2\right)e^1\wedge e^2
+\frac{F\mathcal V}{2}\,e^0\wedge e^3,
\\
\mathcal R_{13}
={}&
-\left(E\mathcal V+\frac{\mathcal V_1}{2}\right)e^0\wedge e^1
-\frac{\mathcal V_2}{2}\,e^0\wedge e^2
-\left(CF+E^2+E_1+\frac{\mathcal V^2}{4}\right)e^1\wedge e^3
\\
&\quad
+\left(BF-EF-E_2\right)e^2\wedge e^3,
\\
\mathcal R_{23}
={}&
\frac{\mathcal V}{2}(C-2F)\,e^0\wedge e^1
-\frac{B\mathcal V}{2}\,e^0\wedge e^2
+\left(CE-EF-F_1\right)e^1\wedge e^3
\\
&\quad
-\left(BE+F^2+F_2\right)e^2\wedge e^3 .
\end{align}
\label{curvatureFormsApp}
\end{subequations}
Equations \eqref{curvatureFormsApp} show explicitly that the curvature contains no uncancelled negative powers of $\Delta$. The apparent divergence of $\omega'(r)$ in \eqref{omegaPrimeRegApp} has disappeared from all orthonormal curvature components. The geometrically defined throat $S_0$ is represented in the one-sided $r$-chart by $r=r_{\rm th}=\sqrt{r_0^2-a^2}$. At the throat,
\begin{equation}
\Sigma_{\rm th}=r_0^2-a^2\sin^2\vartheta
=r_0^2\left(1-\hat a^2q\right),\qquad
\hat a=\frac{a}{r_0},\qquad q=\sin^2\vartheta.
\end{equation}
For $|\hat a|<1$ and $0\leq q\leq1$, one has $1-\hat a^2q\geq1-\hat a^2>0$. Thus, $\Sigma_{\rm th}$ is strictly positive throughout the throat surface for every finite $J_{\rm ADM}$. In particular, all coefficients in \eqref{curvatureFormsApp} have finite one-sided limits at $r=r_{\rm th}$. For example,
\begin{equation}
B\big|_{\rm th}=E\big|_{\rm th}=0,\qquad
B_1\big|_{\rm th}=\frac{r_{\rm th}^2}{\Sigma_{\rm th}^2},\qquad
E_1\big|_{\rm th}=\frac{r_{\rm th}^2}{r_0^2\Sigma_{\rm th}},
\end{equation}
while $\mathcal V$, $\mathcal V_2$, $C$, $C_2$, $F$, and $F_2$ are finite away from the usual polar coordinate degeneracy of the axial frame. The latter is the standard coordinate singularity of spherical polar coordinates. Notice that the scalar contractions below have finite limits also at $q=0$. Since all Riemann-tensor components in a regular orthonormal frame admit finite limits at the throat, every scalar polynomial curvature invariant constructed from contractions of the Riemann tensor is finite there. The following three invariants are displayed as explicit checks:
\begin{equation}
I_1=R,\qquad
I_2=R_{\mu\nu}R^{\mu\nu},\qquad
I_3=R_{\mu\nu\rho\sigma}R^{\mu\nu\rho\sigma}.
\end{equation}
The regularity argument through Eq.~\eqref{curvatureFormsApp} is general and does not use $J_{\rm ADM}=ar_0$. To display compact explicit contractions, we now specialise only this check to the illustrative slice $\hat J=\hat a\equiv j$. On that slice, direct contraction at the throat gives
\begin{subequations}\label{throatInvariantsApp}
\begin{align}
I_1\big|_{\rm th}
&=
\frac{2P_1(j,q)}
{r_0^2\left(1-j^2q\right)^3},
\\
I_2\big|_{\rm th}
&=
\frac{2P_2(j,q)}
{r_0^4\left(1-j^2q\right)^6},
\\
I_3\big|_{\rm th}
&=
\frac{4P_3(j,q)}
{r_0^4\left(1-j^2q\right)^6},
\end{align}
\end{subequations}
where
\begin{align}
P_1(j,q)
={}&
9j^6q^3+j^6q^2-18j^4q^2-2j^4q
+8j^2q+3j^2-1,
\\
P_2(j,q)
={}&
486j^{12}q^6+j^{12}q^4
-1944j^{10}q^5-18j^{10}q^4-5j^{10}q^3
\nonumber\\
&+2916j^8q^4+90j^8q^3+12j^8q^2
-1980j^6q^3-130j^6q^2-13j^6q
\nonumber\\
&+560j^4q^2+58j^4q+9j^4
-36j^2q-8j^2+2,
\\
P_3(j,q)
={}&
891j^{12}q^6-18j^{12}q^5+j^{12}q^4
-3528j^{10}q^5+54j^{10}q^4-6j^{10}q^3
\nonumber\\
&+5148j^8q^4+20j^8q^3+14j^8q^2
-3150j^6q^3-244j^6q^2-14j^6q
\nonumber\\
&+516j^4q^2+298j^4q+9j^4
+124j^2q-118j^2+3 .
\end{align}
Since the denominators in \eqref{throatInvariantsApp} are strictly positive for $0\leq |j|<1$, these invariants are finite everywhere on the throat. In the joint static limit along this slice, $j\to0$, one recovers
\begin{equation}
I_1\big|_{\rm th}=-\frac{2}{r_0^2},
\qquad
I_2\big|_{\rm th}=\frac{4}{r_0^4},
\qquad
I_3\big|_{\rm th}=\frac{12}{r_0^4},
\end{equation}
as expected for the zero-redshift Morris--Thorne geometry. The limiting value $|j|=1$ is singular. On the equatorial section $q=1$, \eqref{throatInvariantsApp} reduces to
\begin{subequations}
\begin{align}
I_1\big|_{\rm th,eq}
&=
-\frac{2(10j^4-10j^2+1)}
{r_0^2(1-j^2)^2},
\\
I_2\big|_{\rm th,eq}
&=
\frac{2(487j^8-993j^6+545j^4-40j^2+2)}
{r_0^4(1-j^2)^4},
\\
I_3\big|_{\rm th,eq}
&=
\frac{4(874j^8-1732j^6+844j^4+12j^2+3)}
{r_0^4(1-j^2)^4}.
\end{align}
\end{subequations}
These expressions diverge as $|j|\to 1$. Therefore, the condition $|a|\leq r_0$ obtained from the coordinate location of the throat should be understood as a reality condition, while the scalar-polynomially regular family satisfies $0\leq |a|<r_0$. Finally, the asymptotic behaviour is regular and agrees with asymptotic flatness
\begin{equation}
I_1=-\frac{2r_0^2}{r^4}+\mathcal O(r^{-6}),
\qquad
I_2=\frac{4r_0^4}{r^8}+\mathcal O(r^{-10}),
\qquad
I_3=\frac{12r_0^4}{r^8}+\mathcal O(r^{-10}) .
\end{equation}
Thus, the spacetime has finite scalar-polynomial curvature invariants at the
throat for $0\leq |a|<r_0$ and approaches flat spacetime at infinity.

\section{Uniqueness of the critical point on the one-parameter slice}\label{app:critical-spin}

In this appendix, we show that the intersection of the general ergoregion curve \eqref{criticalErgoCurve} with the illustrative slice $\hat J=\hat a\equiv j$ is unique. On this slice the critical equation is
\begin{equation}\label{criticalSpinConditionAppendix}
\frac{6}{|j|}\left[K(|j|)-E(|j|)\right]=1.
\end{equation}
We prove that it admits exactly one solution in the interval $0<|j|<1$. Here $K(k)$ and $E(k)$ denote the complete elliptic integrals of the first and second kind, respectively. It is convenient to introduce $k:=|j|$ with $0<k<1$, and to define
\begin{equation}
F(k):=\frac{6}{k}\left[K(k)-E(k)\right].
\end{equation}
The critical equation on the selected slice is $F(k)=1$. We first show that $F(k)$ is strictly increasing on $0<k<1$. Using the standard integral representations
\begin{equation}
K(k)=\int_0^{\pi/2}\frac{d\theta}{\sqrt{1-k^2\sin^2\theta}},\quad
E(k)=\int_0^{\pi/2}\sqrt{1-k^2\sin^2\theta}\,d\theta ,
\end{equation}
we obtain
\begin{equation}
K(k)-E(k)=k^2\int_0^{\pi/2}\frac{\sin^2\theta}{\sqrt{1-k^2\sin^2\theta}}\,d\theta .
\end{equation}
Therefore,
\begin{equation}
\frac{K(k)-E(k)}{k}=k\int_0^{\pi/2}\frac{\sin^2\theta}{\sqrt{1-k^2\sin^2\theta}}\,d\theta .
\end{equation}
Differentiating under the integral sign gives
\begin{equation}
\frac{d}{dk}\left[\frac{K(k)-E(k)}{k}\right]=
\int_0^{\pi/2}\left[\frac{\sin^2\theta}{\sqrt{1-k^2\sin^2\theta}}+\frac{k^2\sin^4\theta}
{(1-k^2\sin^2\theta)^{3/2}}\right]d\theta.
\end{equation}
For $0<k<1$, the denominator is strictly positive on the integration domain, and the integrand is strictly positive for $\theta\in(0,\pi/2]$. Hence
\begin{equation}
    \frac{d}{dk}
    \left[
        \frac{K(k)-E(k)}{k}
    \right]
    >
    0,
    \qquad 0<k<1 .
\end{equation}
It follows that $F'(k)>0$ for $0<k<1$, so the left-hand side of Eq.~\eqref{criticalSpinConditionAppendix} is a strictly increasing function of $k=|j|$. It remains to determine the endpoint behaviour. For \(k\to0\), the complete elliptic integrals admit the expansions \cite{Abramowitz1972}
\begin{equation}
K(k)=\frac{\pi}{2}\left[1+\frac{k^2}{4}+\mathcal{O}(k^4)\right],\quad
E(k)=\frac{\pi}{2}\left[1-\frac{k^2}{4}+\mathcal{O}(k^4)\right].
\end{equation}
Therefore,
\begin{equation}
K(k)-E(k)=\frac{\pi}{4}k^2+\mathcal{O}(k^4),
\end{equation}
and hence
\begin{equation}
F(k)=\frac{3\pi}{2}k+\mathcal{O}(k^3).
\end{equation}
Thus,
\begin{equation}
    \lim_{k\to 0^+}F(k)=0 .
\end{equation}
On the other hand, as $k\to1^{-}$, the complete elliptic integral $K(k)$ diverges logarithmically. This can be seen directly from its integral representation. In the limiting case $k=1$, the integrand becomes $1/\cos{\theta}$, which is non-integrable at the endpoint $\theta=\pi/2$. Equivalently, setting $u=(\pi/2)-\theta$, one has $\cos\theta\sim u$ near the endpoint, and hence the integrand behaves as $1/u$, producing a logarithmic divergence. Thus, $K(k)\longrightarrow +\infty$, as $k\to1^{-}$. By contrast, $E(k)$ remains finite in this limit. Indeed,
\begin{equation}
E(1)=\int_0^{\pi/2}\cos\theta\,d\theta=1.
\end{equation}
Consequently, $K(k)-E(k)\longrightarrow +\infty$, as $k\to 1^{-}$. Since the prefactor $6/k$ tends to the finite positive value $6$, we conclude that
\begin{equation}
\lim_{k\to1^-}F(k)=+\infty .
\end{equation}
We have therefore shown that $F(k)$ is continuous and strictly increasing on $0<k<1$, with
\begin{equation}
\lim_{k\to0^+}F(k)=0,\qquad
\lim_{k\to1^-}F(k)=+\infty.
\end{equation}
By the intermediate value theorem, there exists at least one value $k_{\rm crit}\in(0,1)$ such that $F(k_{\rm crit})=1$. Since $F$ is strictly increasing, this solution is unique. Hence, the critical equation on the selected slice
\begin{equation}
\frac{6}{|j|}\left[K(|j|)-E(|j|)\right]=1
\end{equation}
defines a unique critical value $|j|=j_{\rm crit}$ on the one-parameter slice $\hat J=\hat a$.


\begin{thebibliography}{10}

\bibitem{Morris1988AJP}
M.~S. Morris and K.~S. Thorne, ``Wormholes in spacetime and their use for interstellar travel: A tool for teaching general relativity,'' {\em Am. J. Phys.}, vol.~56, p.~395, 1988.

\bibitem{Ellis1973JMP}
H.~G. Ellis, ``Ether flow through a drainhole - a particle model in general relativity,'' {\em J. Math. Phys.}, vol.~14, p.~104, 1973.

\bibitem{Ellis1974JMP}
H.~G. Ellis, ``Errata: Ether flow through a drainhole: A particle model in general relativity,'' {\em J. Math. Phys.}, vol.~15, p.~520, 1974.

\bibitem{Ellis1979GRG}
H.~G. Ellis, ``The evolving, flowless drain hole: A nongravitating model in general relativity theory,'' {\em Gen. Rel. Grav.}, vol.~10, p.~105, 1979.

\bibitem{Bronnikov1973APP}
K.~A. Bronnikov, ``Scalar-tensor theory and scalar charge,'' {\em Acta Phys. Polon. B}, vol.~4, p.~251, 1973.

\bibitem{Visser1996}
M.~Visser, {\em Lorentzian Wormholes}.
\newblock American Institute of Physics, 1996.

\bibitem{Lobo2008CQGR}
F.~S.~N. Lobo, ``Exotic solutions in general relativity: Traversable wormholes and 'warp drive' spacetimes,'' in {\em Classical and Quantum Gravity Research} (C.~Mikkel and T.~Rasmussen, eds.), pp.~1--78, Hauppauge, NY: Nova Science, 2008.

\bibitem{Bhawal1992PRD}
B.~Bhawal and S.~Kar, ``Lorentzian wormholes in einstein-gauss-bonnet theory,'' {\em Phys. Rev. D}, vol.~46, p.~2464, 1992.

\bibitem{Kanti2011PRL}
P.~Kanti, B.~Kleihaus, and J.~Kunz, ``Wormholes in dilatonic einstein-gauss-bonnet theory,'' {\em Phys. Rev. Lett.}, vol.~107, p.~271101, 2011.

\bibitem{Kanti2012PRD}
P.~Kanti, B.~Kleihaus, and J.~Kunz, ``Stable lorentzian wormholes in dilatonic einstein-gauss-bonnet theory,'' {\em Phys. Rev. D}, vol.~85, p.~044007, 2012.

\bibitem{Antoniou2020PRD}
G.~Antoniou, A.~Bakopoulos, P.~Kanti, B.~Kleihaus, and J.~Kunz, ``Novel einstein–scalar-gauss-bonnet wormholes without exotic matter,'' {\em Phys. Rev. D}, vol.~101, p.~2, 2020.

\bibitem{Gao2017JHEP}
P.~Gao, D.~L. Jafferis, and A.~Wall, ``Traversable wormholes via a double trace deformation,'' {\em JHEP}, vol.~12, p.~151, 2017.

\bibitem{Wall2013CQG}
A.~C. Wall, ``The generalized second law implies a quantum singularity theorem. classical and quantum gravity,'' {\em Class. Quantum Grav.}, vol.~30, p.~165003, 2013.

\bibitem{HochbergVisser1997PRD}
D.~Hochberg and M.~Visser, ``Geometric structure of the generic static traversable wormhole throat,'' {\em Phys. Rev. D}, vol.~56, pp.~4745--4755, 1997.

\bibitem{HochbergVisser1998PRD}
D.~Hochberg and M.~Visser, ``Dynamic wormholes, antitrapped surfaces, and energy conditions,'' {\em Phys. Rev. D}, vol.~58, p.~044021, 1998.

\bibitem{TomikawaIzumiShiromizu2015PRD}
Y.~Tomikawa, K.~Izumi, and T.~Shiromizu, ``New definition of a wormhole throat,'' {\em Phys. Rev. D}, vol.~91, p.~104008, 2015.

\bibitem{Teo1998PRD}
E.~Teo, ``Rotating traversable wormholes,'' {\em Phys. Rev. D}, vol.~58, p.~024014, 1998.

\bibitem{Kuhfittig2003PRD}
P.~K.~F. Kuhfittig, ``Axially symmetric rotating traversable wormholes,'' {\em Phys. Rev. D}, vol.~67, p.~064015, 2003.

\bibitem{Kashargin2008GC}
P.~Kashargin and S.~Sushkov, ``Slowly rotating wormholes: the first order approximation,'' {\em Grav. Cosmol.}, vol.~14, p.~80, 2008.

\bibitem{Kashargin2008PRD}
P.~Kashargin and S.~Sushkov, ``Slowly rotating scalar field wormholes: The second order approximation,'' {\em Phys. Rev. D}, vol.~78, p.~064071, 2008.

\bibitem{Kleihaus2014PRD}
B.~Kleihaus and J.~Kunz, ``Rotating ellis wormholes in four dimensions,'' {\em Phys. Rev. D}, vol.~90, p.~121503(R), 2014.

\bibitem{Hoffmann2018PRD}
C.~Hoffmann, T.~Ioannidou, S.~Kahlen, B.~Kleihaus, and J.~Kunz, ``Symmetric and asymmetric wormholes immersed in rotating matter,'' {\em Phys. Rev. D}, vol.~97, p.~124019, 2018.

\bibitem{Chew2019PRD}
X.~Y. Chew, V.~Dzhunushaliev, V.~Folomeev, B.~Kleihaus, and J.~Kunz, ``Rotating wormhole solutions with a complex phantom scalar field,'' {\em Phys. Rev. D}, vol.~100, p.~044019, 2019.

\bibitem{Azad2023}
B.~Azad, ``Quasinormal modes of static ellis-bronnikov wormholes,'' in {\em Gravity, Cosmology, and Astrophysics: A Journey of Exploration and Discovery with Female Pioneers} (B.~Hartmann and J.~Kunz, eds.), vol.~1022 of {\em Lecture Notes in Physics}, Berlin, Heidelberg: Springer Verlag, 2023.

\bibitem{Azad2024PLB}
B.~Azad, J.~L. Bl$\acute{\mbox{a}}$zquez-Salcedo, F.~S. Khoo, and J.~Kunz, ``Are slowly rotating ellis-bronnikov wormholes stable?,'' {\em Phys. Lett. B}, vol.~848, p.~138349, 2024.

\bibitem{Cisterna2023PRD}
A.~Cisterna, K.~M\"{u}ller, K.~Pallikaris, and A.~Viganò, ``Exact rotating wormholes via ehlers transformations,'' {\em Phys. Rev. D}, vol.~108, p.~024066, 2023.

\bibitem{Clement2023PLB}
G.~Clément and D.~Gal'tsov, ``Rotating traversable wormholes in einstein-maxwell theory,'' {\em Phys. Lett. B}, vol.~838, p.~137677, 2023.

\bibitem{Tanghpati2024NPB}
T.~Tangphati, B.~Chaihao, D.~Samart, P.~Channuie, and D.~Momeni, ``Rotating traversable wormhole geometries in the presence of three-form fields,'' {\em Nucl. Phys. B}, vol.~999, p.~116446, 2024.

\bibitem{Batic2026EPJC}
D.~Batic, D.~Dutykh, and M.~E. Sukaiti, ``Exact spinning morris--thorne wormhole: causal structure, shadows, and multipole moments,'' {\em Eur. Phys. J. C}, vol.~86, p.~179, 2026.

\bibitem{Batic2026CQGSlow}
D.~Batic, D.~Dutykh, and M.~E. Sukaiti, ``Exact solutions for slowly rotating wormholes in the presence of an anisotropic fluid,'' {\em Class. Quantum Grav.}, vol.~43, p.~095007, 2026.

\bibitem{Bambi2013PRD}
C.~Bambi, ``Can the supermassive objects at the centers of galaxies be traversable wormholes? the first test of strong gravity for mm/sub-mm very long baseline interferometry facilities,'' {\em Phys. Rev. D}, vol.~87, p.~107501, 2013.

\bibitem{Bambi2013PRDa}
C.~Bambi, ``Broad $k\alpha$ iron line from accretion disks around traversable wormholes,'' {\em Phys. Rev. D}, vol.~87, p.~084039, 2013.

\bibitem{Nedkova2013PRD}
P.~G. Nedkova, V.~K. Tinchev, and S.~S. Yazadjiev, ``Shadow of a rotating traversable wormhole,'' {\em Phys. Rev. D}, vol.~88, p.~124019, 2013.

\bibitem{Gyulchev2018EPJC}
G.~Gyulchev, P.~Nedkova, V.~Tinchev, and S.~Yazadjiev, ``On the shadow of rotating traversable wormholes,'' {\em Eur. Phys. J. C}, vol.~78, p.~544, 2018.

\bibitem{Shaikh2018PRD}
R.~Shaikh, ``Shadows of rotating wormholes,'' {\em Phys. Rev. D}, vol.~98, p.~024044, 2018.

\bibitem{Deligianni2021PRD}
E.~Deligianni, J.~Kunz, P.~Nedkova, S.~Yazadjiev, and R.~Zheleva, ``Quasiperiodic oscillations around rotating traversable wormholes,'' {\em Phys. Rev. D}, vol.~104, p.~024048, 2021.

\bibitem{Deligianni2021PRDa}
E.~Deligianni, B.~Kleihaus, J.~Kunz, P.~Nedkova, and S.~Yazadjiev, ``Quasiperiodic oscillations in rotating ellis wormhole spacetimes,'' {\em Phys. Rev. D}, vol.~104, p.~064043, 2021.

\bibitem{Kumar2024PRD}
S.~Kumar, A.~Uniyal, and S.~Chakrabarti, ``Shadow and weak gravitational lensing of rotating traversable wormhole in nonhomogeneous plasma spacetime,'' {\em Phys. Rev. D}, vol.~109, p.~104012, 2024.

\bibitem{ChengXuZhao2026EPJC}
P.~Cheng, R.-F. Xu, and P.~Zhao, ``On the cuspy structure of rotating wormhole shadows,'' {\em Eur. Phys. J. C}, vol.~86, p.~371, 2026.

\bibitem{DzhunushalievFolomeev2026GRG}
V.~Dzhunushaliev and V.~Folomeev, ``Rotating wormholes in {Einstein--Dirac--Maxwell} theory,'' {\em Gen. Relativ. Gravit.}, vol.~58, p.~39, 2026.

\bibitem{Azreg2014EPJC}
M.~Azreg-A\"{i}nou, ``From static to rotating to conformal static solutions: rotating imperfect fluid wormholes with(out) electric or magnetic field,'' {\em Eur. Phys. J. C}, vol.~74, p.~2865, 2014.

\bibitem{Mazza2021JCAP}
J.~Mazza, E.~Franzin, and S.~Liberati, ``A novel family of rotating black hole mimickers,'' {\em J. Cosmol. Astropart. Phys.}, vol.~04, p.~082, 2021.

\bibitem{KarJanaKar2025PRD}
A.~Kar, S.~Jana, and S.~Kar, ``New rotating lorentzian wormhole spacetime,'' {\em Phys. Rev. D}, vol.~111, p.~064010, 2025.

\bibitem{Abramowitz1972}
M.~Abramowitz and I.~A. Stegun, {\em Handbook of Mathematical Functions}.
\newblock Dover: New York, 1972.

\bibitem{Wald1986}
R.~M. Wald, {\em General Relativity}.
\newblock The University of Chicago Press, 1986.

\bibitem{bray1986kerr}
I.~Bray, ``Kerr black hole as a gravitational lens,'' {\em Phys. Rev. D}, vol.~34, no.~2, p.~367, 1986.

\end{thebibliography}

\end{document}